\documentclass[10pt]{IEEEtran}
\usepackage{setspace}
\usepackage{graphicx}
\usepackage{amsmath, amsthm}
\usepackage{cases}
\usepackage{stfloats}
\usepackage{subfigure}
\usepackage{graphics,epsf}
\usepackage{amssymb}
\usepackage{float}
\usepackage{dsfont}
\usepackage[bookmarks=false]{hyperref}
\usepackage{setspace}

\newtheorem{theorem}{Theorem}
\newtheorem{lemma}{Lemma}

\newcommand{\co}{\"{o}}

\begin{document}
\title{A Secure Communication Game with a Relay Helping the Eavesdropper}
\author{Melda~Yuksel,~\IEEEmembership{Member,~IEEE,}
        ~Xi~Liu,~\IEEEmembership{Student~Member,~IEEE}
        and~Elza~Erkip,~\IEEEmembership{Fellow,~IEEE}
\thanks{Copyright (c) 2010 IEEE. Personal use of this material is permitted. However, permission to use this material for any other purposes must be obtained from the IEEE by sending a request to pubs-permissions@ieee.org.}

\thanks{Manuscript received September 15, 2010.; revised January 20, 2011. This material is based upon work partially supported by the National Science Foundation under Grant No.0635177 and by the New York State Center for Advanced Technology in Telecommunications (CATT). The material in this paper was presented in part at the IEEE Information Theory Workshop, Taormina, Italy, October 2009, and at the 21st Annual IEEE International Symposium on Personal, Indoor and Mobile Radio Communications, Istanbul, Turkey, September 2010.}

\thanks{M. Yuksel with the Electrical and Electronics Engineering Department, TOBB
University of Economics and Technology, Ankara, Turkey (e-mail:
yuksel@etu.edu.tr).

X. Liu and E. Erkip are with the Electrical and Computer Engineering Department,
Polytechnic Institute of New York University, Brooklyn, NY 11201 USA (e-mail: xliu02@students.poly.edu,
elza@poly.edu).}}\maketitle

\begin{abstract}
In this work a four terminal complex Gaussian network composed of a source, a destination, an eavesdropper and a jammer relay is studied under two different set of assumptions: (i) The jammer relay does not hear the source transmission, and (ii) The jammer relay is causally given the source message. In both cases the jammer relay assists the eavesdropper and aims to decrease the achievable secrecy rates. The source, on the other hand, aims to increase it. To help the eavesdropper, the jammer relay can use pure relaying and/or send interference. Each of the problems is formulated as a two-player, non-cooperative, zero-sum continuous game. Assuming Gaussian strategies at the source and the jammer relay in the first problem, the Nash equilibrium is found and shown to be achieved with mixed strategies in general. The optimal cumulative distribution functions (cdf) for the source and the jammer relay that achieve the value of the game, which is the Nash equilibrium secrecy rate, are found. For the second problem, the Nash equilibrium  solution is found and the results are compared to the case when the jammer relay is not informed about the source message.
\end{abstract}

\section{Introduction}



In wireless communications, messages are broadcasted, and any
transmission can be overheard by nearby nodes. If eavesdroppers are present in the environment, then all confidential
information become vulnerable and can be identified. Therefore, security against
eavesdropping is an essential system requirement for all wireless
communication applications.

In addition to eavesdropping, wireless networks are also prone to jamming. In contrast to eavesdropping, jamming is an active attack, in which deliberate signals are transmitted to prevent proper reception at the intended receiver. Node capture attacks can also take place to compromise message confidentiality \cite{LuQianChen07,TaguePoovendran08}.


Security against eavesdropping using information theoretic principles was first considered in \cite{Shannon49}. In \cite{Wyner75}, the wire-tap channel was studied
for the degraded case, when the eavesdropper's received signal is a degraded version of the legitimate receiver's observation. This model was extended to less noisy and more capable wire-tap channels
in \cite{CsiszarK78}. The Gaussian wire-tap channel was studied in \cite{LeungYCHellman78}.



Unconditional security for the relay channel with an external eavesdropper is investigated in \cite{LaiEG06_2}, \cite{YukselE07} and \cite{AggarwalSCP08}. In \cite{LaiEG06_2}, the authors suggest the noise forwarding scheme, where the relay transmits dummy codewords that can be decoded at the destination. While sending dummy codewords does not hurt the legitimate communication, it increases the confusion at the eavesdropper, and hence helps achieve a higher secrecy rate. Noise forwarding scheme is similar to cooperative jamming \cite{TekinYener06_2}, in which one of the users in the system injects noise to increase achievable secrecy rates in multi-access and two-way channels. The paper \cite{TangLiuSpasojevicPoor07_2} ties \cite{LaiEG06_2} and \cite{TekinYener06_2} together, and shows that the relay can choose between sending structured codewords and pure noise to increase achievable secrecy rates even further. When the relay sends dummy codewords or forwards noise, the gains in achievable secrecy rates are due to the \emph{interference} the relay creates. Thus, in the rest of the paper, we will collect both schemes under the name \emph{interference assistance} as in \cite{TangLiuSpasojevicPoor07_2}.

In this paper we consider a four terminal complex Gaussian network with a source-destination pair, a relay and an external eavesdropper. Unlike the above mentioned works, in which relay's transmissions aim to help the legitimate source-destination communication, we assume that the relay is captured by an adversary and aims to help the eavesdropper instead of the source-destination communication. Thus, we refer to the relay node as the \emph{jammer relay}. The jammer relay is capable of helping the eavesdropper as well as jamming the destination to reach its objective of smaller secrecy rates.

Reliable communication in the presence of arbitrary jamming strategies and no eavesdropper is in general a complex problem \cite{LapidothNarayan98}. In some special cases, the optimal jammer strategy as well as the optimal encoding at the source node can be solved. For example \cite{KashyapBasarSrikant04} solves for optimal transmitter and jammer strategies under a game-theory framework when the jammer is informed about the transmitter's signal. Reference \cite{ShafieeUlukus05} investigates the multiple access channel with a correlated jammer. In \cite{WangGiannakis08} authors study a game between a jammer and a relay, where the relay assists the source-destination communication.


In our preliminary work we analyzed the source and the jammer relay transmitting in
orthogonal separate time slots \cite{YukselE07_2}. Since the jammer relay is
malicious, the legitimate receiver can choose not to listen to it. Hence in such a scenario interference assisted eavesdropping is not a possibility, the jammer relay does pure relaying. The eavesdropper listens to both the source and the jammer relay to decrease its
equivocation.

%

More generally, the transmissions of the source and the jammer relay are not orthogonal and the jammer relay may acquire some knowledge of the source signal. Then both interference assisted eavesdropping and pure relaying are possible. When the jammer relay can only operate in a half-duplex fashion, there is a tradeoff  between pure relaying and interference assisted eavesdropping. Pure relaying can either try to convey source information to the eavesdropper or attempt to cancel the source signal at the destination by jamming but the relay has to listen the source first and can only participate in transmission a fraction of the time. In interference assisted eavesdropping, the jammer relay does not utilize the overheard signal and can simultaneously transmit with the source, but its signal does not carry useful information about the source message and only contains interference to confuse the destination.

While solving the half-duplex jammer relay problem seems difficult, as a first step to address the general case, in this paper we investigate two extreme situations: (i) Problem 1- The jammer relay does not hear the source
transmission and can therefore transmit simultaneously with the source, and (ii) Problem 2- The jammer relay knows the source signal causally. In the first problem, pure relaying is not an option, and only interference assisted eavesdropping protocols are meaningful. On the other hand, both pure relaying and interference assisted eavesdropping can be useful in the second problem.

As the jammer relay and the source have conflicting interests about the value of the secrecy rate, we formulate each of the problems as a two-player non-cooperative zero-sum continuous game and find achievable secrecy rates for the source-destination communication. Assuming Gaussian strategies at the source and the jammer relay in the first problem, we calculate the Nash equilibrium and show that it is achieved with mixed strategies in general. We also state the optimal cumulative distribution functions (cdf) for the source and the jammer relay that achieve the value of the game, which is the Nash equilibrium secrecy rate. For the second problem, we find the Nash equilibrium solution and compare the results to the case when the jammer relay is not informed about the source message.

In the next section the general system model is described. In Section~\ref{sec:Problem1} we solve the first problem, in which the jammer relay does not hear the source transmission. In Section~\ref{sec:Problem2} we attack the second problem assuming the jammer relay is given the source signal causally. In Section~\ref{sec:NumRes}, numerical results are presented and in Section~\ref{sec:Conclusion} we conclude.

\section{General System Model}\label{sec:systemmodel}

We investigate the four terminal complex Gaussian network composed of a source, a
destination, an eavesdropper and a jammer relay denoted by S, D, E and JR
respectively. The network under investigation is shown in Fig.~\ref{fig:systemmodel}.

The received signals at the destination and the eavesdropper are
\begin{eqnarray}
Y_{D,i} & = &  h_{SD}X_{S,i} + h_{RD}X_{R,i} + Z_{D,i} \label{eqn:YD}\\
Y_{E,i} & = &  h_{SE}X_{S,i} + h_{RE}X_{R,i} + Z_{E,i}, \label{eqn:YE}
\end{eqnarray}
where $X_{S,i}$ and $X_{R,i}$ are the signals the source and the jammer
relay transmit at time $i$, $i=1,...,n$. In the first problem under study, the jammer relay does not hear the source signal. In the second problem, it is assumed that the jammer relay is given the source signal causally, and hence $X_{R,i}$ depends on $X_{S,1}^i$, where $X_{S,1}^i = (X_{S,1},...,X_{S,i})$. The complex channel gains between node $k$ and node $l$ are denoted as $h_{kl}$, $k = {S,R}$, $l ={D,E}$. All channel gains are fixed and assumed to be known at all nodes. The complex additive Gaussian noises at the destination and at the
eavesdropper are respectively denoted as $Z_{D,i}$ and $Z_{E,i}$
and are independent and identically distributed (i.i.d.) with
zero mean and variance $N_D = N_E = N_0$. The source and the jammer relay have
average power constraints $P_S$ and $P_R$. For convenience we will write $\gamma_{kl}  =  {|h_{kl}|^2}{P_k}/N_0$, $k = {S,R}$, $l ={D,E}$, to indicate the received power at node $l$
due to node $k$.

The source aims to send the message $W$ securely to the destination in $n$ channel uses. The secrecy rate, $R_s$ is defined as the maximum information rate such that the secrecy constraint is satisfied; i.e. $\lim_{n\rightarrow \infty}H(W)/n = \lim_{n\rightarrow \infty}H(W|Y_{E,1}^{n})/n$, and the probability of decoding error at the destination approaches zero as $n$ approaches infinity \cite{Wyner75}.

In this problem the source and the jammer relay have opposing interests. The former wants to increase the secrecy rate, and the latter wants the decrease it. Thus, this problem constitutes a zero-sum game, where the utility is the secrecy rate, $R_s$. The source and the jammer relay make their decisions simultaneously, and hence the game is strategic.

If the jammer relay does not exist we have the Gaussian eavesdropper channel, for which sending Gaussian codewords at full power is optimal \cite{LeungYCHellman78}. If the eavesdropper does not exist ($h_{SE}=h_{RE}=0$), then the jammer relay's only objective is to jam the destination to decrease the information rate. This problem is solved in \cite{KashyapBasarSrikant04}, for which it is shown that correlated jamming is optimal. The optimal transmission strategy for the source is to send Gaussian codewords, and the optimal strategy for the jammer is of the form
\begin{eqnarray}
X_{R,i} = \rho X_{S,i} + Z_i, \label{eqn:JRstrategy}
\end{eqnarray}  where $\rho \in \mathds{C}$, $Z_i$ is independent of $X_{S,i}$ and chosen i.i.d. according to distribution $\mathcal{CN}(0,N_Z)$. If the jammer has enough power, then it can completely block the source-destination communication. If its power is not large enough, it shares its power in canceling the source message and sending noise.

For the two settings considered in this paper, the strategy spaces for both the source and the jammer relay can be quite large, and finding the Nash equilibrium solution is very complicated in general. However, the results of \cite{KashyapBasarSrikant04} and \cite{TangLiuSpasojevicPoor07_2} suggest that both correlated jamming and sending structured codewords at the jammer relay have a high potential to decrease achievable secrecy rates. Therefore, we assume that the jammer relay strategies have the same form as in (\ref{eqn:JRstrategy}). Different than \cite{KashyapBasarSrikant04} to incorporate interference assisted jamming/relaying we assume that in (\ref{eqn:JRstrategy}) the signal $Z_1^n$ can be a structured codeword or Gaussian noise. If the jammer relay were a jammer only, then $Z_1^n$ would simply be unstructured noise. However, in this paper it is not merely a jammer, but aims help the eavesdropper to ensure no secret communication takes place. Unstructured noise is useful as it harms the legitimate communication, yet it also harms the eavesdropper. Structured codewords have the potential to help the eavesdropper more as they can be potentially decoded at the eavesdropper.

In Problem 1, as the relay does not hear the source transmission, we have $\rho = 0$. In Problem 2, the level of source and jammer relay signal correlation can be adjusted as a function of $\rho$. Depending on the value of $\rho$, the jammer relay can try to enable decoding at the eavesdropper as in \cite{YukselE07_2}, or it can attempt to cancel $X_S$ at the destination. Under these assumptions, we are interested in finding a scheme attaining the Nash equilibrium of the game, which consists of a pair of optimal strategies for the source and the jammer relay.

\section{Problem 1: The relay does not hear the source}\label{sec:Problem1}

In this section we state the game theoretic formulation for the first problem. Solving this game, we suggest an achievability scheme that results in the Nash equilibrium value of the game. In Sections~\ref{subsec:GameTheoryForm} to \ref{subsec:OtherCases}, we assume all nodes in the system know both the source and the jammer relay's strategies, including the codebooks. In Section \ref{subsec:unknownJRcodebook}, we will extend this study to the more realistic case where the destination does not know the jammer relay codebook. Throughout Section \ref{sec:Problem1}, we assume the source and the jammer relay operate at full power.

\subsection{Game Theoretic Formulation}\label{subsec:GameTheoryForm}

In the first problem, the relay does not hear the source node. Therefore, its strategy cannot depend on the source signal and $\rho = 0$ in (\ref{eqn:JRstrategy}). Then, as described in Section \ref{sec:systemmodel}, the jammer relay either generates dummy codewords or simply forwards noise. When both the source and the jammer relay send structured codewords, the problem becomes similar to a multiple-access channel with an external eavesdropper. In a multiple-access channel, both transmitters need to be decoded at the destination. However, in this game, the jammer relay only sends dummy codewords, and does not need to be decoded either at the destination or at the eavesdropper.

Observe that when the jammer relay sends complex Gaussian codewords/noise with full power, the best distribution for choosing the source codebook is the complex Gaussian distribution~\cite{LiangPS06} with zero mean and variance $P_S$. On the other hand, when the source sends complex Gaussian codewords with zero mean and variance $P_S$, then the jammer relay distribution that decreases the secrecy rate the most is the complex Gaussian distribution with zero mean and variance $P_R$ \cite{DiggaviCover01}. Motivated by these observations, we assume that the source and the jammer relay choose their codebooks according to i.i.d. complex Gaussian with zero mean and variances $P_S$ and $P_R$ respectively. Then, the source strategy is to choose the rate of the information it wants to convey to the destination, and the jammer relay strategy is to choose the rate of its dummy information. In the remainder of this section the source strategy $\xi$ will denote the rate of information, while the jammer relay strategy $\eta$ will denote the rate of dummy information. We argue below that structured codewords for sending dummy information at the jammer relay also include the possibility of sending pure noise. Given source and jammer relay strategies, the secrecy rate, or the payoff, is a function of both $\xi$ and $\eta$. If a certain positive secrecy rate, $R_s(\xi,\eta)$\footnotemark \footnotetext{For Problem 1 the payoff is a function of $\xi$ and $\eta$, whereas in Section \ref{sec:Problem2}, where Problem 2 is discussed, the payoff will be defined as a function $\rho$ and $N_Z$ defined in (\ref{eqn:JRstrategy}).} is achieved, then the source node's payoff is equal to $R_s(\xi,\eta)$ and the relay's payoff is equal to $-R_s(\xi,\eta)$.

Under these assumptions, the destination can decode both the source and the jammer relay codewords if the rate pair $(\xi,\eta)$ is in $\mathcal{R}^{{[D]}}_{\mathrm{MAC}}$
\begin{eqnarray}
\mathcal{R}^{{[D]}}_{\mathrm{MAC}}= \left\{(\xi,\eta)\left|\begin{array}{rcc}
  \xi &\leq & \log(1+\gamma_{SD})\\
\eta &\leq & \log(1+\gamma_{RD})\\
\xi +\eta &\leq & \log(1+\gamma_{SD}+\gamma_{RD})
\end{array}\right.\right\}.\label{eqn:R-D-MAC}
\end{eqnarray}However, the jammer relay only sends dummy codewords, and does not need to be decoded either at the destination or at the eavesdropper. If the destination cannot decode the jammer relay codeword, it can simply treat it as noise. Thus, all $\xi$ rates in $\mathcal{R}^{[D]}_{\mathrm{N}}$
\begin{eqnarray}
\mathcal{R}^{[D]}_{\mathrm{N}}= \left\{\xi \left|\begin{array}{rcc}
  \xi &\leq & \log\left(1+\frac{\gamma_{SD}}{1+\gamma_{RD}}\right)
\end{array}\right.\right\} \label{eqn:R-D-N}
\end{eqnarray}
are achievable as well. Overall, we say that the destination can decode the source information with arbitrarily small probability of error, if $(\xi,\eta) \in \mathcal{R}^{[D]}$ where
\begin{eqnarray} \mathcal{R}^{[D]} = \mathcal{R}^{{[D]}}_{\mathrm{MAC}} \bigcup \mathcal{R}^{{[D]}}_{\mathrm{N}} \label{eqn:R-D}\end{eqnarray} Note that after taking the union, the individual constraint on $\eta$ in (\ref{eqn:R-D-MAC}) is not needed anymore. We also define two other regions $\mathcal{R}^{{[E]}}_{\mathrm{MAC}}$, $\mathcal{R}^{[E]}_{\mathrm{N}}$, as in (\ref{eqn:R-D-MAC}) and (\ref{eqn:R-D-N}) replacing all $D$ with $E$. Then \begin{eqnarray} \mathcal{R}^{[E]} = \mathcal{R}^{{[E]}}_{\mathrm{MAC}} \bigcup \mathcal{R}^{{[E]}}_{\mathrm{N}} \label{eqn:R-E}\end{eqnarray}
Then for a fixed source and jammer relay rate pair $(\xi,\eta)$, the payoff function, $R_s(\xi,\eta)$, is equal to
\begin{eqnarray}
\lefteqn{R_s(\xi,\eta)} \nonumber\\
&=& \left\{\begin{array}{l}
              0,  \mathrm{~~~~if~}  \left(\begin{array}{l}
                                   (\xi,\eta) \in \mathcal{R}^{[E]}  \mathrm{~or} \\
                                  (\xi,\eta) \not\in \mathcal{R}^{[D]}
                                 \end{array}\right) \\
             \displaystyle{\max_{\nu}}~ (\xi - \nu),  \mathrm{~if~}  \left(\begin{array}{l}
                                              (\xi, \eta) \in \mathcal{R}^{[D]}  \mathrm{~~~~and} \\
                                              (\nu,\eta) \not\in \mathcal{R}^{[E]}  
                                            \end{array}\right)\end{array}
 \right.
 \label{eqn:Rs}
\end{eqnarray}
The proof of how this secrecy rate would be achieved is similar to \cite{TangLiuSpasojevicPoor07_2} and is skipped here.

An example is shown in Fig.~\ref{fig:DifferentCases} for the boundaries of the regions $\mathcal{R}^{[D]}$ and $\mathcal{R}^{[E]}$ with the corner points defined as
\begin{eqnarray}
(\Delta_S,\Delta_R) & = & \left(\log(1+\gamma_{SD}),\log\left(1+\frac{\gamma_{RD}}{1+\gamma_{SD}}\right)\right) \label{eqn:cornerpoints1}\\
(\Omega_S,\Omega_R) & = & \left(\log\left(1+\frac{\gamma_{SD}}{1+\gamma_{RD}}\right), \log(1+\gamma_{RD})\right) \\
(\delta_S,\delta_R) & = & \left(\log(1+\gamma_{SE}), \log\left(1+\frac{\gamma_{RE}}{1+\gamma_{SE}} \right)\right)\allowdisplaybreaks\\
(\omega_S,\omega_R) & = & \left(\log\left(1+\frac{\gamma_{SE}}{1+\gamma_{RE}}\right), \log(1+\gamma_{RE})\right).
\label{eqn:cornerpoints4}
\end{eqnarray}

For a fixed $(\xi,\eta)$, the secrecy rate defined in (\ref{eqn:Rs}) corresponds to the horizontal distance between the point $(\xi,\eta)$  and the dashed line in Fig.~\ref{fig:DifferentCases}, if $(\xi,\eta)$ is in between the solid and dashed lines. If $(\xi,\eta) \in \mathcal{R}^{[E]}$, that is inside the dashed line, then both the destination and the eavesdropper can reliably decode the source information, and the secrecy rate is zero. If $(\xi,\eta) \not\in \mathcal{R}^{[D]}$, outside the solid line, the destination cannot decode the source message reliably. The secrecy rate is zero, because there is no reliable communication between the source and the destination. Because of this immediate drop in secrecy rates beyond the boundary of $\mathcal{R}^{[D]}$, the payoff function is discontinuous.

Note that choosing the dummy information rate as $\eta = \omega_R$ is equivalent to sending unstructured Gaussian noise at the jammer relay. Thus, the secrecy rate achieved by jammer relay sending unstructured Gaussian noise is also covered in our model, although we arrived at the $R_s(\xi,\eta)$ function assuming structured codewords for the jammer relay.

Depending on $\gamma_{kl}$, the positions of the corner points with respect to each other change, and multiple cases arise. In the next subsection we investigate the case, where the conditions
\begin{eqnarray}
 \log(1+\gamma_{SE}+\gamma_{RE}) & \leq & \log(1+\gamma_{SD}+\gamma_{RD}) \label{eqn:cond1}\\
 \delta_S  \leq \Delta_S, &\quad& \omega_S  \leq \Omega_S \leq \delta_S\\
\Delta_R  \leq \delta_R, &\quad& \delta_R  \leq \Omega_R \leq \omega_R \label{eqn:cond5}
\end{eqnarray}
are all satisfied. The case shown in Fig.~\ref{fig:DifferentCases} satisfies all these conditions. We will call this \emph{Case A}. There are 10 other cases B-M shown in Fig. \ref{fig:casea_m} for $\Delta_S > \delta_S$ and $\Omega_S>\omega_S$, and Case N for $\Delta_S \leq \delta_S$ or $\Omega_S\leq \omega_S$. These other cases require similar techniques and are explained in Section~\ref{subsec:OtherCases}.

\subsection{Solution to the Game: Case A}\label{sec:PlayingGame}

A zero-sum game has a pure strategy solution if
\[ \max_{{\xi}} \min_{{\eta}} R_s(\xi,\eta) = \min_{{\eta}} \max_{{\xi}} R_s(\xi,\eta). \] However, if there is no $(\xi,\eta)$ that satisfies this equation, then no pure strategy Nash equilibrium exists, and a mixed strategy solution is needed \cite{Karlin}.

\begin{lemma} When $\gamma_{kl}$, $k = S,R$, $l= D,E$ satisfy the conditions (\ref{eqn:cond1})-(\ref{eqn:cond5}), which define Case A, the two-player zero-sum game does not have a pure strategy solution. \label{thm:Pure}
\end{lemma}
\begin{proof}In this game
$\max_{{\xi}} \min_{{\eta}} R_s(\xi,\eta) = 0,$ whereas $\min_{{\eta}} \max_{{\xi}} R_s(\xi,\eta) = \log(1+\gamma_{SD}+\gamma_{RD})- \log(1+\gamma_{SE}+\gamma_{RE}).$ These two values are not the same, hence a pure strategy solution does not exist.
\end{proof}

\begin{lemma}\label{lemma:2} The game defined in Lemma~\ref{thm:Pure}, is equivalent to a continuous game played over the square, where the source and jammer relay strategies are respectively restricted to the compact intervals $\xi \in [\Omega_S,\Omega_S+L]$ and $\eta \in [\delta_R,\delta_R +L]$, where $L $ is the edge length $L =  \Omega_R - \delta_R$. \label{thm:GameReduction}
\end{lemma}
\begin{proof}To prove this we eliminate comparable and inferior strategies for the source and the jammer relay. First note that
\[ R_s(\xi,\eta) \leq  R_s(\Omega_S,\eta), \mathrm{~for~}\xi<\Omega_S, \mathrm{~and~}0 < \eta. \] In other words, as both players are rational, the source never chooses rates $\xi < \Omega_S$. Similarly, the source node never chooses its rate larger than $\Delta_S$, as the secrecy rate $R_s(\xi,\eta) = 0$ no matter what the jammer relay action is. On the other hand,
\[ R_s(\xi,\eta) \geq R_s(\xi, \Omega_R),\mathrm{~for~} \Omega_S \leq \xi \leq \Delta_S, \mathrm{~and~}\eta>\Omega_R. \] For the jammer relay, choosing any rate larger than $\Omega_R$ is inferior to choosing rate equal to $\Omega_R$ and thus we can omit the strategies $\eta > \Omega_R$. Similarly, \[ R_s(\xi,\eta) \geq R_s(\xi, \delta_R),\mathrm{~for~}\Omega_S \leq \xi \leq \Delta_S, \mathrm{~and~}\eta<\delta_R. \] The jammer relay strategies $\eta <\delta_R$ are inferior to $\eta = \delta_R$ and thus the jammer relay never chooses its rate less than $\delta_R$. Finally, in this reduced game, the source node does not choose its rate larger than $\log(1+\gamma_{SD}+\gamma_{RD})-\delta_R$, as this choice makes its payoff equal to zero. In other words,
\begin{eqnarray*}
 R_s(\xi ,\eta) &\leq &  R_s(\log(1+\gamma_{SD}+\gamma_{RD})-\delta_R,\eta), \end{eqnarray*} for $\xi > \log(1+\gamma_{SD}+\gamma_{RD})-\delta_R$ and $\delta_R<\eta<\Omega_R$. These strategy eliminations result in the desired reduced game.
\end{proof} The reduced game based on Lemma~\ref{thm:GameReduction} is shown in Fig.~\ref{fig:ReducedCaseA}.

We next describe how to solve for the value of this reduced game.
\begin{theorem}\label{thm:Solution} Let $a$ be defined as $a  = (\delta_S-\Omega_S)/{L}.$ Suppose $a \in [k/(k+1),(k+1)/(k+2)]$, for some integer $k\geq 0$. Then the equivalent game in Lemma~\ref{thm:GameReduction} has the Nash equilibrium secrecy rate $R_s^* = L\alpha (1-a)$, where $\alpha = g_k(a)$, and is achieved with cdfs for the source and the jammer relay $F_{\xi}(\xi)$ and $F_{\eta}(\eta)$, respectively. Here the functional forms of $\alpha$, $F_{\xi}(\xi)$ and $F_{\eta}(\eta)$ depend\footnote{The dependency of the cdfs on $k$ is not explicitly shown for notational convenience.} on $k$, and for a given $k$, both $\alpha$ and the cdfs can be readily computed. For example, for $0 \leq a \leq 1/2$ or equivalently $k=0$, we have $\alpha = g_0(a) = \frac{e^{-1/(1-a)}}{1-\frac{a}{1-a}e^{-1}}$ and
\begin{equation}
F_\xi(\xi) = \left\{
\begin{array}{l }
\alpha e^{\frac{\xi-\Omega_S}{L(1-a)}} ~~~~~~~~~~ \Omega_S \leq \xi \leq \Omega_S + L(1-a)\\
\alpha \left[(1+e^{-1})e^{\frac{\xi-\Omega_S}{L(1-a)}}-\frac{1}{1-a}(\frac{\xi-\Omega_S}{L}) e^{\frac{\xi-\Omega_S}{L(1-a)}-1}\right] \nonumber \\~~~~~~~~~~~~~~~~~~~~~\Omega_S+ L(1-a)\leq \xi \leq \Omega_S + L
\end{array}
\right. \label{eqn:cdf0-a-0.5}
\end{equation}

The optimal cdf for the jammer relay $F_\eta(\eta)$ is the same as $F_\xi(\xi)$ if $\xi$ and $\Omega_S$ are replaced with $\eta$ and $\delta_R$ respectively.
\end{theorem}
\begin{proof}
See Appendix A.
\end{proof}

Since there are infinitely many intervals for $a$ in Theorem~\ref{thm:Solution} (corresponding to each nonnegative integer $k$), it is important to find a practical way to calculate the value of the game.
\begin{theorem}\label{thm:Discrete} Consider a discrete approximation of the equivalent game in Lemma~\ref{thm:GameReduction} obtained by dividing the square into a uniform grid of $(T+1)^2$ samples. The discrete source strategies are $\xi_i = \omega_S + Li/T$, the relay strategies are $\eta_j = \delta_R + Lj/T$, and the payoff matrix is $A = [a_{ij}]$, where $a_{ij} = R_s(\xi_i,\eta_j)$, $i,j = 0,1,...,T$. The value of this discrete game can be obtained using linear programming. Furthermore, for a chosen T, difference between values of the discrete and the continuous game is at most $2\sqrt{2}L/T$.
\end{theorem}
\begin{proof}
See Appendix B.
\end{proof}

To compare the optimal strategies for the continuous and discrete games using the solutions in Theorems~\ref{thm:Solution} and \ref{thm:Discrete}, we assume $|h_{SD}| = 1$, $|h_{RD}| = 1/2$, and $|h_{SE}| = |h_{RE}| = 2/3$. We make no assumptions on the phases of the channel gains as $F_{\xi}(\xi)$, $F_{\eta}(\eta)$ and $R_s^*$ only depend on the magnitude of the channel gains. The source and the jammer relay power constraints are $P_S = P_R = 10$. We then have $a = 0.5255$, $L = 0.946$ and $\alpha = 0.20484$. We find the Nash equilibrium secrecy rate of the continuous game as 0.092 bits/channel use. To use the discrete approximation we set $T = 400$. Choosing this sample size, Theorem \ref{thm:Discrete} states that the difference between the value of discrete and continuous games is at most $0.007$. Yet, the actual difference is much smaller and we find that the value of the discrete game as 0.0923 bits/channel use. Note that these values are much smaller than the no jammer relay case, for which the secrecy rate is equal to 1.0146 bits/channel use. The optimal cdfs for the source for the continuous and discrete games are very close to each other and are shown in Fig.~\ref{fig:CDFcaseAa1}. As argued in Lemma \ref{thm:GameReduction}, we observe that $F_{\xi}(\xi)$ is zero if $\xi < 1.947 = \Omega_S$, and is 1 if $\xi > 2.893= \Omega_S+L$. Note that in the reduced game, sending Gaussian noise with full power is still one of jammer relay's possible strategies.

\subsection{Note on the Achievability of the Mixed Strategy Solution}\label{subsec:Achievability}

If the solution of the game is a pure strategy, the achievability follows using the arguments in Section~\ref{subsec:GameTheoryForm}. As the solution in Theorem~\ref{thm:Solution} is mixed, it is also important to explain how the Nash equilibrium is attained information theoretically.

In a mixed strategy the players randomize their actions over a set of strategies with a certain probability distribution. The players act repeatedly and ignore any strategic link that may exist between plays. They also know each other's probability distribution functions, and hence formulate their own actions. In the game defined in this section, when a mixed strategy solution is needed, the source node assumes a variable rate scheme, similar to the one adopted for fading eavesdropper channels~\cite{GopalaLEG06}.

In this variable rate scheme, the source generates a total of
$2^{nB {E}(\xi)}$ codewords, where $B$ is the number of blocks the game is played where each block is of length $n$, and $E(\xi)$ is the expected rate for the source node, expectation calculated over the joint cdf $F_{\xi}(\xi)F_{\eta}(\eta)$.
The source uses these codewords to form a secure code that conveys $nBR_s^*$ bits of information in $B$ blocks \cite{Wyner75}, where $R_s^*$ is the value of the game or the Nash secrecy rate. In each block, the source independently chooses a rate $\xi$ according to $F_{\xi}(\xi)$ and transmits $n\xi$ bits of the codeword chosen to represent the secure information. Similarly, the jammer relay chooses a rate $\eta$ according to $F_{\eta}(\eta)$. Since the eavesdropper cannot improve its mutual information more than $\xi$, as in the variable rate case of \cite{GopalaLEG06}, (\ref{eqn:Rs}) is still valid and  $R_s^*$ is attained as both $n$ and $B$ approach infinity.

\subsection{Other Cases: B-N}\label{subsec:OtherCases}

So far in Sections \ref{subsec:GameTheoryForm}/\ref{sec:PlayingGame}, we have obtained a complete solution for Case A. In this subsection, we will show that for all other cases B to N, shown in Fig. \ref{fig:casea_m}, either an analytical solution or a discrete approximation can be found. Due to limited space, the coordinates of corner points of the equivalent regions are not shown in the subfigures of Fig. \ref{fig:casea_m} but they can be easily determined given those coordinates of corner points on the boundary regions as in (\ref{eqn:cornerpoints1})-(\ref{eqn:cornerpoints4}).

For Cases B and C, following a reasoning similar to Lemma \ref{thm:GameReduction} the reduced region remains to be a square as in Case A except that the jammer relay rate $\eta$ in interval $[\omega_R,\Omega_R)$ is dominated by the jammer relay rate $\Omega_R$. In this case, in order to keep the support of the jammer relay strategy compact, we do not eliminate these dominated jammer relay rates. Using the same reasoning as in the proof of Theorem \ref{thm:Discrete}, we can solve the problem by approximating the original game with a discretized matrix game.

For Case D, the solution is the same as Case A when $a =0$ and the optimal solution can be obtained analytically using Theorem~\ref{thm:Solution}. In this case, the Nash equilibrium secrecy rate $R_s^* = L\alpha + (\Omega_S - \delta_S)$, where $L = \Omega_R - \delta_R$ and $\alpha = e^{-1}$.

For Case E, a pure strategy Nash equilibrium exists, the optimal strategies for the source and the jammer relay are $(\Omega_S,\Omega_R)$ and the Nash equilibrium secrecy rate is $\Omega_S - \delta_S$.

For Cases F-K, if we retain some dominated relay rates to keep the support of the relay strategy compact, the reduced region becomes a rectangle as shown in Fig. ~\ref{fig:casea_m}. We can easily extend Theorem 1 to the case of a rectangle and show that these games also have a value. Furthermore, similar to Theorem \ref{thm:Discrete}, a discrete approximation can be computed.

For Case L, after all eliminations similar to Lemma~\ref{thm:GameReduction}, only four points remain. The two rates for the source to choose are $\xi = \Omega_S$ and $\xi =\Delta_S$ while those for the jammer relay are $\eta= \delta_R$ and $\eta = \Omega_R$. The $2\times 2$ matrix game with the source and the jammer relay being the row and column players respectively has the following payoff matrix
\begin{equation}
\left(\begin{matrix}0 & \Omega_S - \omega_S\\ \Delta_S - \delta_S & 0 \end{matrix}\right).
\end{equation}
A mixed-strategy Nash equilibrium exists, in which the source chooses $\xi = \Omega_S $ with probability $\frac{\Delta_S - \delta_S}{\Omega_S - \omega_S + \Delta_S - \delta_S}$ and $\xi = \Delta_S$ with probability $\frac{\Omega_S - \omega_S}{\Omega_S - \omega_S + \Delta_S - \delta_S}$ while the jammer relay chooses $\eta = \delta_R$ with probability $\frac{\Omega_S - \omega_S}{\Omega_S - \omega_S + \Delta_S - \delta_S}$ and $\eta = \Omega_R$ with probability $\frac{\Delta_S - \delta_S}{\Omega_S - \omega_S + \Delta_S - \delta_S}$. The Nash equilibrium secrecy rate is given by $R_s^* = \frac{(\Omega_S - \omega_S)(\Delta_S - \delta_S)}{\Omega_S - \omega_S + \Delta_S - \delta_S}.$

Case M is very similar to Case L. As in Case L, the source can choose between $\xi = \Omega_S$ and $\xi =\Delta_S$ while the jammer relay can choose between $\eta= \delta_R$ and $\eta = \Omega_R$. The $2\times 2$ matrix game with the source and the jammer relay being the row and column players respectively has the following payoff matrix
\begin{equation}
\left(\begin{matrix}\Omega_S-\delta_S & \Omega_S - \omega_S\\ \Delta_S - \delta_S & 0 \end{matrix}\right).
\end{equation}
A mixed-strategy Nash equilibrium also exists, in which the source chooses $\xi = \Omega_S $ with probability $\frac{\Delta_S - \delta_S}{\Delta_S - \omega_S }$ and $\xi = \Delta_S$ with probability $\frac{\delta_S-\omega_S}{\Delta_S - \omega_S}$ while the jammer relay chooses $\eta = \delta_R$ with probability $\frac{\Omega_S - \omega_S}{\Delta_S - \omega_S}$ and $\eta = \Omega_R$ with probability $\frac{\Delta_S - \Omega_S}{\Delta_S - \omega_S}$. The Nash equilibrium secrecy rate is given by $R_s^* = \frac{(\Omega_S - \omega_S)(\Delta_S - \delta_S)}{\Delta_S - \omega_S}.$

Case N is the case when boundaries of the regions $\mathcal{R}^{[D]}$ and $\mathcal{R}^{[E]}$ intersect or $\mathcal{R}^{[D]}$ is contained in $\mathcal{R}^{[E]}$. In the former scenario, the intersection point is a pure strategy Nash equilibrium leading to zero secrecy rate; in the latter case, no positive secrecy rate can be achieved regardless of what the source and jammer relay strategies are.

\subsection{Unknown Jammer Relay Codebook}\label{subsec:unknownJRcodebook}

In previous subsections, it is assumed that all nodes in the system know both the source and the jammer relay strategies, including the codebooks. In practice, it is reasonable to assume that the destination is not aware of the jammer relay's exact codebook even though it may have the knowledge that the jammer relay uses Gaussian codebooks. Under this assumption, the destination cannot jointly decode the source and the jammer relay messages any more and thus the jammer relay's signal would always be treated as noise at the destination. Hence in this case, the destination can correctly decode source information only if $\xi \in \mathcal{R}^{[D]}_{\mathrm{N}}$, where $\mathcal{R}^{[D]}_{\mathrm{N}}$ is given in (\ref{eqn:R-D-N}). Redrawing the rate regions under the new assumption, we obtain the two cases shown in Fig. \ref{fig:region_combined} for $\Omega_S \leq \delta_S$ and for  $\Omega_S > \delta_S$. It can be easily shown that if $\Omega_S \leq \delta_S$, the intersection point, $(\xi, \eta) = (\Omega_S, \frac{(\omega_R-\delta_R)\Omega_S+\omega_S\delta_R-\omega_R\delta_S}{\omega_S-\delta_S})$, is a pure Nash equilibrium point and is unique, leading to zero secrecy rate. If $\Omega_S > \delta_S$ instead, all points $(\Omega_S, \eta)$ with $\eta \in [0,\delta_R]$ are Nash equilibria, and the resulting Nash equilibrium secrecy rate is $\Omega_S - \delta_S$. Therefore, we conclude that when the destination does not know the jammer relay codebook, the Nash equilibrium secrecy rate is $R_s = (\Omega_S - \delta_S)^+$, where $[x]^+ = \max(0,x)$.

\section{Problem 2: The Relay is Given the Source Signal}\label{sec:Problem2}

In this section we provide the Nash equilibrium solution of the game for the second problem in which the jammer relay is given the source signal causally. To further simplify the setup we assume the destination is unaware of the jammer relay codebook as in Section~\ref{subsec:unknownJRcodebook}. Depending on whether $Z_1^n$ is Gaussian noise or structured codeword in (\ref{eqn:JRstrategy}), the payoff function and the Nash equilibrium of the game are different. In the following, we consider these two different scenarios respectively.

\subsection{$Z_1^n$ is Gaussian noise}\label{subsec:Problem2noise}

In this subsection, we assume that $Z_i$ in (\ref{eqn:JRstrategy}) are i.i.d. complex Gaussian with zero mean and variance $N_Z$, for all $i=1,...,n$. Substituting $X_{R,i}$ in (\ref{eqn:JRstrategy}) into (\ref{eqn:YD}) and (\ref{eqn:YE}), the relation between the source signal and the received signals at the destination and the eavesdropper can equivalently be written as
\begin{eqnarray}
\tilde{Y}_{D,i} = X_{S,i} + \tilde{Z}_{D,i} \label{eqn:YDtilde}\\
\tilde{Y}_{E,i} = X_{S,i} + \tilde{Z}_{E,i}\label{eqn:YEtilde}
\end{eqnarray}
where $\tilde{Z}_{D,i} = \frac{h_{RD}Z_i+Z_{D,i}}{h_{SD}+h_{RD}\rho}$ and $\tilde{Z}_{E,i} = \frac{h_{RE}Z_i+Z_{E,i}}{h_{SE}+h_{RE}\rho}$. Since $Z_D,\: Z_E \sim \mathcal{CN}(0,N_0)$, the variances of $\tilde{Z}_D$ and $\tilde{Z}_E$ can respectively be obtained as
\begin{eqnarray}
\tilde{N}_D = \frac{|h_{RD}|^2 N_Z + N_0}{|h_{SD}+h_{RD}\rho|^2}, \label{eqn:tildeND}\\
\tilde{N}_E = \frac{|h_{RE}|^2 N_Z + N_0}{|h_{SE}+h_{RE}\rho|^2}. \label{eqn:tildeNE}
\end{eqnarray}

The system in (\ref{eqn:YDtilde}) and (\ref{eqn:YEtilde}) is equivalent to a Gaussian wire-tap channel. Hence the best strategy for the source is to choose its codebook according to i.i.d. $\mathcal{CN}(0, P_S)$, resulting in the secrecy capacity $R_s$ \cite{LeungYCHellman78}
\begin{equation}
R_s(\rho,N_Z) = \left[\log_2\left(1+\frac{P_S}{\tilde{N}_D}\right) - \log_2\left(1+\frac{P_S}{\tilde{N}_E}\right)\right]^+.
\end{equation}
On the other hand, if the source fixes its input distribution to be complex Gaussian with zero mean and variance $P_S$, the jammer relay would intend to choose $\rho$ and $N_Z$ to minimize the secrecy rate when its strategy is limited to (\ref{eqn:JRstrategy}). Thus, the pure-strategy achieving the Nash equilibrium for the source is to choose the input distribution $X_S \sim \mathcal{CN}(0,P_S)$ and for the jammer relay is to transmit $X_R = \rho X_S + Z$ for some optimal $\rho$ and $N_Z$, where the optimal $\rho$ and $N_Z$ are denoted as $\rho^*$ and $N_Z^*$, are solutions to the problem
\begin{equation}
\min R_s(\rho,N_Z) \quad \text{subject to } |\rho|^2 P_S + N_Z \leq P_{R}. \label{eqn:OptProb}
\end{equation}
Then the Nash equilibrium secrecy rate is $R_s^*(\rho^*,N_Z^*)$.

The optimization problem in (\ref{eqn:OptProb}) is not convex. This suggests infeasibility of a closed form solution in general. However, for the following special cases, the solutions have a simple form and can be obtained as follows:
\begin{enumerate}
\item If $|h_{SE}|\geq |h_{SD}|$, even without any help from the jammer relay for the eavesdropper, the secrecy rate is always zero. Thus, the jammer relay only needs to keep silent.
\item If $|h_{RD}|\geq \sqrt{P_S/P_{R}}|h_{SD}|$, then the pair $(\rho^*,N_Z^*) = (-h_{SD}/h_{RD}, 0)$ is optimal. When the link between the jammer relay and the destination is strong, the jammer relay can send a negatively correlated signal to completely cancel the source signal.
\item Suppose $|h_{SE}|<|h_{SD}|$ and $|h_{RD}|<\sqrt{P_S/P_{R}}|h_{SD}|$. If $|h_{RD}h_{SE}|>|h_{SD}h_{RE}|$ and $\frac{N_0(|h_{SD}|^2-|h_{SE}|^2)}{|h_{RD}h_{SE}|^2-|h_{SD}h_{RE}|^2}<P_{R}$, then
$\rho^* = 0$, $N_Z^* =  \frac{N_0(|h_{SD}|^2-|h_{SE}|^2)}{|h_{RD}h_{SE}|^2-|h_{SD}h_{RE}|^2}$ are optimal. This is the case when the jammer relay is capable of forcing secrecy rate to zero by transmitting only noise.
\end{enumerate}

When $|h_{SE}|<|h_{SD}|$ and $|h_{RD}|<\sqrt{P_S/P_{R}}|h_{SD}|$, numerical methods are used to solve the optimization problem of (\ref{eqn:OptProb}). In (\ref{eqn:OptProb}), it can be shown that the constraint $|\rho|^2 P_S + N_Z \leq P_{R}$ is not necessarily met with equality; i.e. it is possible that in the optimal solution the jammer relay should not transmit with full power. This is in contrast to the jamming problem without an eavesdropper, where it is best for the jammer to use full power \cite{KashyapBasarSrikant04}. Letting $\rho = |\rho|e^{j\theta}$, $N_Z = wP_S$, the constraint can be rewritten as
\begin{eqnarray}
|\rho| \leq \sqrt{P_{R}/P_{S}},\\
0 \leq \theta \leq 2\pi,\\
0 \leq w \leq P_{R}/P_{S}-|\rho|^2
\end{eqnarray}
To numerically solve for (\ref{eqn:OptProb}) we exhaustively search over all feasible $|\rho|$, $\theta$ and $w$ in the range.

\subsection{$Z_1^n$ is structured codeword}\label{subsec:Problem2structured}

In this subsection, we assume $Z_1^n$ is a structured codeword instead of noise. As in Section~\ref{subsec:unknownJRcodebook}, we assume the jammer relay shares its codebook only with the eavesdropper, but not with the adversarial destination. Note that, this set of assumptions define the worst case scenario along with the fact that the relay is informed about the source information causally.

The jammer relay is capable of choosing the rate of the codebook of $Z_1^n$ so that the eavesdropper can successfully decode it. After the eavesdropper decodes the dummy information carried by $Z_1^n$, it subtracts $Z_1^n$ and re-scales its received signal $Y_E$ to get
\begin{equation}
\bar{Y}_{E,i}= X_{S,i} + \bar{Z}_{E,i}
\end{equation}
where $\bar{Z}_{E,i} = \frac{Z_E}{h_{SE}+h_{RE}\rho}$ has zero mean and variance
\begin{equation}
\bar{N}_E = \frac{N_0}{|h_{SE}+h_{RE}\rho|^2}.\label{eqn:barNE}
\end{equation}
Since the destination does not know the jammer relay codebook, it can only treat $Z_1^n$ as noise. If the jammer relay's transmitted signal has the form in (\ref{eqn:JRstrategy}) and $Z_1^n$ is chosen to be Gaussian, then based on the results in \cite{LeungYCHellman78}, the optimal source distribution is $\mathcal{CN}(0,P_S)$. On the other hand, if the source distribution is fixed to $\mathcal{CN}(0,P_S)$, it is best for the jammer relay to construct codebook of $Z_1^n$ according to the distribution $\mathcal{CN}(0,N_Z)$ since only the destination is affected by $Z_1^n$ and the worst noise is the Gaussian one \cite{DiggaviCover01}. Hence, the pure-strategy achieving the Nash equilibrium for the source is to use $\mathcal{CN}(0,P_S)$ distribution to generate the source codebook while that of the jammer relay is to construct $Z_1^n$ as Gaussian codewords and transmit signal of the form in~(\ref{eqn:JRstrategy}) for some optimal $\rho$ and $N_Z$. Given that both $X_S$ and $Z$ are Gaussian at the equilibrium, the jammer relay ought to set the rate of the codebook of $Z$ to satisfy
\[R_Z < \log_2 \left(\frac{|h_{RE}|^2N_Z}{N_0+|h_{SE}+h_{RE}\rho|^2P_S}\right),\] in order to guarantee that the eavesdropper can decode it successfully.
With $X_S \sim \mathcal{CN}(0,P_S)$ and $X_R = \rho X_S + Z$, the secrecy rate can be expressed as
\begin{equation}
R_s(\rho, N_Z) = \left[\log_2\left(1+\frac{P_S}{\tilde{N}_D}\right) - \log_2\left(1+\frac{P_S}{\bar{N}_E}\right)\right]^+. \label{eqn:OptProb2}
\end{equation}

Comparing the equivalent noise at the eavesdropper $\bar{N}_E$ defined in (\ref{eqn:barNE}) and $\tilde{N}_E$ defined in (\ref{eqn:tildeNE}), we observe that the Nash equilibrium secrecy rates obtained when $Z_1^n$ is a structured codeword will always be less than or equal to the Nash equilibrium secrecy rates obtained when $Z_1^n$ is Gaussian noise. However, structured $Z_1^n$ requires the jammer and the eavesdropper to share the codebook information in advance and having $Z_1^n$ as Gaussian noise leads to a simpler system design.

To find the optimal $\rho$ and $N_Z$ at the Nash equilibrium, we need to minimize $R_s(\rho, N_Z)$ subject to $\rho^2P_S+N_Z \leq P_{R}$. As in the previous subsection, a closed form solution cannot be obtained in general. For some special cases, we can easily find an optimal solution. We list a few of these cases as follows.
\begin{enumerate}
\item If $|h_{SE}| \geq |h_{SD}|$, then it is enough for the jammer relay to keep silent.
\item If $|h_{RD}|\geq \sqrt{P_S/P_{R}}|h_{SD}|$, then $\rho^* = -h_{SD}/h_{RD}$ and $N_Z^* = 0$ are optimal.
\item If $|h_{SE}|<|h_{SD}|$ and $|h_{RD}|<\sqrt{P_S/P_{R}}|h_{SD}|$, and
if $\frac{|h_{SD}|^2-|h_{SE}|^2}{|h_{SE}h_{RD}|^2}\leq P_{R}/N_0$, then $(\rho^*,\allowbreak N_Z^*) \allowbreak = \allowbreak (0,\frac{|h_{SD}|^2-|h_{SE}|^2}{|h_{SE}h_{RD}|^2}N_0)$ is optimal.
\end{enumerate}

As in the previous subsection, we can use numerical methods to optimize the function in (\ref{eqn:OptProb2}) by exhaustively searching over all $|\rho|$'s, $\theta$'s and $w$'s.

\section{Numerical results}\label{sec:NumRes}

In this section we present some numerical results to show how the secrecy rate changes with the jammer relay-eavesdropper channel quality, when the jammer relay is given the source signal causally (Problem 2). We also compare the secrecy rates attained when $Z_1^n$ is Gaussian noise, Section \ref{subsec:Problem2noise} and $Z_1^n$ is a structured codeword, Section \ref{subsec:Problem2structured}. We finally provide comparisons of the secrecy rates for Problem 1 and Problem 2.

For Problem 2, when $Z_1^n$ is Gaussian noise, the Nash equilibrium secrecy rate is plotted as a function of $h_{RE}$ in Fig.~\ref{fig:Fig2} for $h_{SD} = 1$, $h_{SE} = 0.4 + 0.4j$ and $h_{RD} = 0.2 - 0.2j$ when $P_S = P_{R} = 10$ and $N_0 = 1$. For simplicity, we restrict $h_{RE}$ to be real in this plot. We note that in general multiple optimal choices for ($\rho^*, N_Z^*$) appear in general. The following two cases are also included in the figure for comparison: 1) the jammer relay only sends Gaussian noise; i.e. $X_R = Z$; 2) the jammer relay only sends a correlated version of source message; i.e. $X_R = \rho X_S$. When $h_{RE}$ is near $0.2$, $X_R = Z$ is close to optimal; when $h_{RE}$ is greater than $0.3$ or smaller than $-0.15$, the curve for $X_R =\rho X_S$ overlaps with the curve for $X_R = \rho X_S+Z$, which implies that for $h_{RE}$'s in these ranges the component of Gaussian noise $Z$ in the signal $X_R$ is not necessary and sending a correlated version of source message is already optimal. Also, we observe that as $h_{RE}$ grows sufficiently large, the secrecy rate $R_s^*$ drops to zero when the jammer relay's signal is of the form $X_R = \rho X_S +Z$.

For Problem 2, Fig.~\ref{fig:Fig6} shows the Nash equilibrium secrecy rate as a function of real $h_{RE}$ when $Z_1^n$ is a structured codeword, $h_{SD} = 1$, $h_{SE} = 0.4+0.4j$ and $h_{RD} = 0.2-0.2j$. As before, $P_S = P_{R} = 10$ and $N_0 = 1$. In Fig.~\ref{fig:Fig6}, for comparison the Nash equilibrium secrecy rate when the jammer relay signal is restricted to $X_R = Z$ with $Z_1^n$ being a structured codeword and $X_R = \rho X_S$ are also included. The secrecy rate for $X_R = \rho X_S$ is the same as in Fig.~\ref{fig:Fig2}, since $X_R$ does not depend on $Z$. Also the secrecy rate for $X_R = Z$ remains constant as $h_{RE}$ changes, which is due to the fact that $\tilde{N}_D$ and $\bar{N}_E$ do not depend on $h_{RE}$ when $\rho = 0$. However, in this case $h_{RE}$ has an impact on the maximum possible rate $R_Z$.
Under the same setting as in Fig.~\ref{fig:Fig2} and Fig.~\ref{fig:Fig6}, Fig.~\ref{fig:Fig7} compares the Nash equilibrium secrecy rates under four different scenarios: 1) Problem 1: uninformed relay with the jammer relay codebook known at the destination, Section \ref{subsec:GameTheoryForm}/\ref{subsec:OtherCases}; 2) Problem 1: uninformed relay with the jammer relay codebook unknown at the destination, Section \ref{subsec:unknownJRcodebook}; 3) Problem 2: informed relay with $Z_1^n$ being Gaussian noise, Section \ref{subsec:Problem2noise}; 4) Problem 2: informed relay with $Z_1^n$ being structured codewords, Section \ref{subsec:Problem2structured}. Our observations can be listed as follows:
\begin{itemize}
\item When the relay is uninformed of source message and the jammer relay codebook is unknown at the destination, as in Section \ref{subsec:unknownJRcodebook}, neither $\Omega_S$ nor $\delta_S$ depend on $h_{RE}$, and the secrecy rate remains unchanged for all $h_{RE}$'s. As the destination always treats the jammer relay's signal as noise while the eavesdropper may decode it, this scenario is actually equivalent to that of $X_R = Z$ in Fig.~\ref{fig:Fig6}.
\item When the relay is uninformed of the source message and the destination has the knowledge of the jammer relay codebook the Nash equilibrium secrecy rate can be increased for $|h_{RE}|<0.6$ compared with the no codebook case at the destination.
\item Among all four scenarios, an informed relay with $Z_1^n$ being structured codeword achieves the smallest secrecy rate.
\item For a range of $h_{RE}$'s studied, an informed relay with $Z_1^n$ as Gaussian noise can achieve a smaller secrecy rate than an uninformed relay with jammer relay codebook unknown at the destination. However, we observe that an informed relay with $Z_1^n$ as Gaussian noise results in the largest secrecy rate when $h_{RE}$ is between $0.15$ and $0.4$.
\end{itemize}

In Fig.~\ref{fig:Fig7}, for simplicity, we only considered real channel gain $h_{RE}$ between the jammer relay and the eavesdropper. In Fig.~\ref{fig:Fig8}, we fix $|h_{RE}| = 0.25$ and plot the Nash equilibrium secrecy rates as a function of the phase of $h_{RE}$, $\angle h_{RE}$ for the four different cases mentioned above. When the relay is uninformed of the source message, irrespective of whether the destination knows the jammer relay codebook or not, the Nash equilibrium secrecy rate only depends on the magnitude of $h_{RE}$ and therefore remains the same for all phases of $h_{RE}$. However, the Nash equilibrium secrecy rates in both cases of informed relay are sensitive to the changes in $\angle h_{RE}$. For example, as shown in Fig.~\ref{fig:Fig8}, when $0 \leq \angle h_{RE}< \pi$, the Nash equilibrium secrecy rate in either case decreases as $\angle h_{RE}$ increases and finally drops to zero when $\angle h_{RE}$ grows beyond a threshold.

\section{Conclusion}\label{sec:Conclusion}

In this paper a four terminal network with a source, a destination, an eavesdropper and a jammer relay is investigated. The source and the jammer relay have conflicting interests. The former aims higher secrecy rates, whereas the latter aims lower secrecy rates. Due to this conflict, this problem is formulated as a non-cooperative two-player zero-sum continuous game. Two different cases are studied: 1) the jammer relay does not hear the source, and 2) the jammer relay is given the source signal causally. For the first case, it is discussed that interference assistance is the only option. Under this assumption the optimal solution for the source and the jammer relay is found to be mixed strategies. The Nash equilibrium secrecy rate of the game is calculated, in addition to optimal cumulative distribution functions for the source and the jammer relay. A discrete approximation to the continuous game, whose value can be made arbitrarily close to the value of the continuous game, is also suggested. For the second case, limiting the jammer relay strategies to a combination of pure relaying and interference assisted eavesdropping schemes, the Nash equilibrium of the game is found and an achievability scheme is suggested. Our results show that the presence of the jammer relay decreases the secrecy rates significantly. If the jammer relay is informed of the source signal, the secrecy rates are even lower. Future work includes the more general half-duplex jammer relay case, in which the relay is not given the source signal for free, but has to listen to it to be able to perform pure relaying.

\appendices
\section{Proof of Theorem 1}
For the equivalent game in Fig. \ref{fig:ReducedCaseA}, the achieved secrecy rate is given by
\begin{equation}
R_s(\xi,\eta) = \left\{
\begin{array}{l}
0, \quad \xi + \eta > \Omega_S + \delta_R + L\\
\xi + \eta - \Omega_S - \delta_R - aL,  \\
~~~~~ \Omega_S + \delta_R + aL<\xi + \eta \leq \Omega_S+\delta_R+L\\
0, \quad \xi + \eta \leq \Omega_S+\delta_R+aL
\end{array}
\right..
\label{eqn:secrate1}
\end{equation}

We first assume $L = 1$ and $(\Omega_S,\delta_R) = (0,0)$ and then use the game-theoretic techniques in \cite{Karlin} to solve the continuous game played over the unit square. For convenience, we denote the jammer relay's pure strategy by $\lambda = 1 - \eta$ instead of $\eta$. Fig.\ref{fig:unitsquare} illustrates the unit square where the normalized game is played. Rewriting (\ref{eqn:secrate1}), we get
\begin{equation}
R_s(\xi,\lambda) = \left\{
\begin{array}{l l}
M_1(\xi,\lambda) &\quad \lambda \geq \xi\\
M_2(\xi,\lambda)&\quad \lambda < \xi
\end{array}
\right.
\label{eqn:kernel_unitsquare}
\end{equation}
where
\begin{equation}
M_1(\xi,\lambda) = \left\{
\begin{array}{l l}
0 &\quad \lambda > \xi + 1 - a\\
\xi-\lambda + 1 - a &\quad \xi \leq \lambda \leq \xi+1-a
\end{array}
\right.,
\label{eqn:M1}
\end{equation}
and
\begin{equation}
M_2(\xi,\lambda) = 0
\label{eqn:M2}
\end{equation}
$M_1(\xi,\lambda)$ and $M_2(\xi,\lambda)$ are defined over the closed triangles $\Sigma_1 = \{(\xi,\lambda)|0\leq \xi \leq \lambda \leq 1\}$ and $\Sigma_2 = \{(\xi,\lambda)|0\leq \lambda \leq \xi \leq 1\}$ respectively as shown in Fig.~\ref{fig:unitsquare}. Here $a$ is a constant between 0 and 1. In game theory, the function $R_s(\xi,\lambda)$ is called the \emph{kernel} of the game.

For the game in Fig. \ref{fig:unitsquare}, the source and the jammer relay strategies can be represented by the cdfs $F_\xi(\xi)$ and $F_\lambda(\lambda)$ defined on $[0,1]$ respectively. Given source strategy $F_\xi(\xi)$ and jammer relay strategy $F_\lambda(\lambda)$, the expected pay-off of the source is given by $\int_0^1\int_0^1 R_s(\xi,\lambda)dF_\xi(\xi)dF_\lambda(\lambda)$.

A solution to the game with kernel $R_s(\xi,\lambda)$ \cite{Karlin} is a pair of cdfs $F_\xi$ and $F_\lambda$ together with a real number $R_s^*$ such that
\begin{equation}
\int_0^1 R_s(\xi,\lambda)dF_\xi(\xi) \geq R_s^*  \quad \mbox{ for all $0\leq\lambda\leq 1$}
\end{equation}
and
\begin{equation}
\int_0^1 R_s(\xi,\lambda)dF_\lambda(\lambda) \leq R_s^* \quad \mbox{ for all $0\leq\xi\leq 1$}
\end{equation}
Here we call $R_s^*$ the value of the game or Nash equilibrium secrecy rate and $F_\xi$ and $F_\lambda$ optimal strategies for the source and the jammer relay respectively.

Motivated by the results on \emph{games of timing} in \cite{Karlin}, we first assume that the optimal strategies for the source and the jammer relay are of the forms $F_{\xi} = (\alpha I_0,f_{\xi})$ and $F_{\lambda} = (f_{\lambda},\beta I_1)$. Here $(\alpha I_0,f_{\xi})$ denotes a distribution made up of a density function $f_{\xi}$ spread on the interval $(0,1]$ and one jump of magnitude $\alpha$ located at $0$. Similarly, $(f_{\lambda},\beta I_1)$ is a distribution made up of a density function spread on the interval $[0,1)$ and one jump of magnitude of $\beta$ located at $1$. Then, we solve the game analytically under these assumptions. After a particular solution is found, we verify that it is indeed a solution of the game. As the kernel $R_s(\xi,\lambda)$ depends on the parameter $a$, we will find optimal strategies for both players for different parameters of $a$.

For $0 \leq a \leq 1/2$ (i.e., $k=0$) mentioned in Theorem \ref{thm:Solution}, suppose $0\leq \lambda \leq 1-a$, then
\begin{eqnarray}
\int_0^1 R_s(\xi,\lambda)dF_{\xi}(\xi) & = & \alpha R_s(0,\lambda) + \int_{0^+}^{\lambda}M_1(\xi, \lambda)f_{\xi}(\xi)d\xi \nonumber\\
 &=&  \alpha (1-a-\lambda) \nonumber \\
 &&{+}\: \int_{0^+}^{\lambda}(\xi-\lambda+1-a)f_{\xi}(\xi)d\xi
\end{eqnarray}
Differentiating w.r.t. $\lambda$ and equating to zero, we get
\begin{equation}
G(\xi) = \alpha(e^{\frac{\xi}{1-a}}-1), \quad 0\leq\xi\leq 1-a
\end{equation}
where $G(\xi) = \int_{0^+}^\xi f_\xi(u)du$.
Similarly, for $\lambda$ such that $1-a<\lambda\leq 1$, we get
\begin{equation}
G(\xi) = -\alpha + C\exp(\frac{\xi}{1-a}) - \frac{\alpha}{1-a}\xi e^{\frac{\xi}{1-a}-1},\quad 1-a<\xi \leq 1
\end{equation}
Due to continuity of $G(\xi)$ at $\xi = 1-a$, we have $C = \alpha(1+e^{-1})$.
Also, since $G(1) = 1 - \alpha$, we obtain $\alpha = g_0(a) = \frac{e^{-1/(1-a)}}{1-\frac{a}{1-a}e^{-1}}$.
Therefore, the optimal strategy for the source, $F_\xi$, has the form:
\begin{equation}
F_\xi(\xi) =
\left\{
\begin{array}{l l}
\alpha &\quad \mbox{if $\xi = 0$}\\
\alpha + G(\xi) &\quad \mbox{if $0<\xi\leq 1$}
\end{array}
\right..
\label{eqn:cdf_xi}
\end{equation}

In the same manner, the optimal strategy for the jammer relay, $F_\lambda$, can be obtained
\begin{equation}
F_\lambda(\lambda) =
\left\{
\begin{array}{l l}
1-\alpha-G(1-\lambda) &\quad\mbox{if $0\leq \lambda < 1$}\\
1 &\quad \mbox{if $ \lambda = 1$}
\end{array}
\right..
\label{eqn:cdf_lambda}
\end{equation}
Alternatively, when $\eta$ is used to denote the jammer relay's pure strategy, the jammer relay's optimal strategy can be expressed as $F_{\eta}(\eta) = 1 - F_\lambda(1-\eta) = F_{\xi}(\eta)$.

It can be readily verified that $
\int_0^1 R_s(\xi,\lambda)dF_\xi(\xi) = \alpha(1-a)$, for all $0\leq\lambda\leq 1$. Also, $
\int_0^1 R_s(\xi,\lambda)dF_\lambda(\lambda) = \alpha(1-a)$, for all $0\leq\xi\leq 1$. Hence, when $0\leq a\leq 1/2$, $F_\xi$ and $F_\lambda$ (or $F_\eta$) in (\ref{eqn:cdf_xi}) and (\ref{eqn:cdf_lambda}) are the optimal strategies for the source and the jammer relay and the secrecy rate is $R_s^* = \alpha(1-a)$.

Now suppose $a \in (k/(k+1),(k+1)/(k+2)]$ for $k>0$. In this case, the optimal strategies can be obtained using the same method as above except that the resulting $F_\xi(\xi)$ is not differentiable at points $\xi =0, 1- a, 2(1-a),..., (k+1)(1-a)$. Hence the functional forms of $\alpha$, $F_{\xi}(\xi)$ and $F_{\eta}(\eta)$ depend on $k$, and for a given $k$ both $\alpha$ and the cdfs can be readily derived.

For the more general case when the edge is equal to $L = \Omega_R - \delta_R$ and the left lower corner point of the square is located at $(\Omega_S, \delta_R)$, we have $ a = \frac{\delta_S-\Omega_S}{L}$. Similar to the above discussion, optimal cdfs and Nash equilibrium secrecy rate can be derived to result in Theorem \ref{thm:Solution}.

\section{Proof of Theorem 2}
The proof of this theorem is based on Theorem 8 in Chapter 17 of \cite{SzepForgo} and the approximation techniques suggested there.

Let us first consider the game over the unit square when $L = 1$ and $(\Omega_S,\delta_R) = (0,0)$. According to Theorem 8 in Chapter 17 of \cite{SzepForgo}, the game with the kernel function $R_s(\xi,\lambda)$ in (\ref{eqn:kernel_unitsquare}) has a value and there exists a pair of ``$\epsilon$-optimal'' strategies for the source and the jammer relay. \footnotemark\footnotetext{A pair of cdfs $F_{\xi}^0(\xi)$ and $F_{\lambda}^0(\lambda)$ are said to be $\epsilon$-optimal and the approximate value of the game $v_\epsilon$ is said to be ``$\epsilon$ good'', if (i) the expected payoff $v_1$ calculated for $F_{\xi}^0(\xi)$ and for any $F_{\lambda}(\lambda)$ satisfies $v_1 \geq v_{\epsilon} - \epsilon$, (ii) the expected payoff $v_2$ calculated for any $F_{\xi}(\xi)$ and for $F_{\lambda}^0(\lambda)$ satisfies $v_2 \leq v_{\epsilon} - \epsilon$, and (iii) $|v_\epsilon - v| \leq \epsilon$, where $v$ is the value of the game \cite{SzepForgo}.}
Hence, it is possible to solve the game using approximate methods. To obtain the approximation, we divide the unit interval by $T-1$ inner grid points equally spaced in it so that the square is divided into a uniform grid of $(T+1)^2$ samples. Therefore, the source chooses among discrete pure strategies $\xi_i = i/T$ ($i = 0,1,...,T$) while the jammer relay chooses among $\lambda_j = j/T$ ($j = 0,1,...,T$). The payoff matrix\footnote{Note that if we use the kernel $R_s(\xi,\eta)$ instead of $R_s(\xi,\lambda)$, the payoff matrix would become the one in Theorem 2.} is $A = [a_{ij}]$, where $a_{ij} = R_s(\xi_i,\lambda_j)$, $i,j = 0,1,...,T$. For this discrete game, a mixed-strategy Nash equilibrium always exists and its value can be obtained using linear programming.

From \cite{SzepForgo}, in order for the equilibrium strategies of the discrete game to be $\epsilon$-optimal for the original continuous game, $T$ needs to be chosen such that $T \geq \frac{2\max [K_1, K_2]}{\epsilon}$, where $K_1$ and $K_2$ are two constants chosen to satisfy
\begin{eqnarray}
\begin{array}{l}
  |M_1(\xi,\lambda)-M_1(\xi,\lambda)|\leq K_1 |(\xi,\lambda) - (\xi',\lambda')| \nonumber\\
 ~~~~~~~~~~~~~~~~~~~~~~~~~~~~~~~ \text{for all}\; (\xi,\lambda)\in \Sigma_1,\; (\xi',\lambda')\in \Sigma_1, \\
 |M_2(\xi,\lambda)-M_2(\xi,\lambda)|\leq K_2 |(\xi,\lambda) - (\xi',\lambda')|\nonumber\\
 ~~~~~~~~~~~~~~~~~~~~~~~~~~~~~~~ \text{for all}\; (\xi,\lambda)\in \Sigma_2,\; (\xi',\lambda')\in \Sigma_2,
\end{array}
\end{eqnarray}

$M_1$ and $M_2$ are as defined in (\ref{eqn:M1}) and (\ref{eqn:M2}).
In our problem, the above conditions hold for $K_1 = \sqrt{2}$ and $K_2 = 0$, therefore we only need to choose $T$ to be greater than $\frac{2\sqrt{2}}{\epsilon}$. On the other hand, if $T$ is fixed, difference between values of the discrete game and the continuous game is at most $2\sqrt{2}/T$. For example, if we set $\epsilon = 0.00708$, it suffices to choose $T$ to be 400.

Similar to Appendix A, for the more general case when the edge is equal to $L = \Omega_R - \delta_R$ and the left lower corner point of the square is located at $(\Omega_S, \delta_R)$, a discrete game can be used to approximate the continuous game as in Theorem \ref{thm:Discrete}. The difference between the values of the two games would be scaled by $L$; i.e., within $2\sqrt{2}L/T$.

\nocite{YukselLiuE09}\nocite{YukselLiuE10}


\begin{figure}[t]
\begin{center}
\includegraphics[width= 3 in]{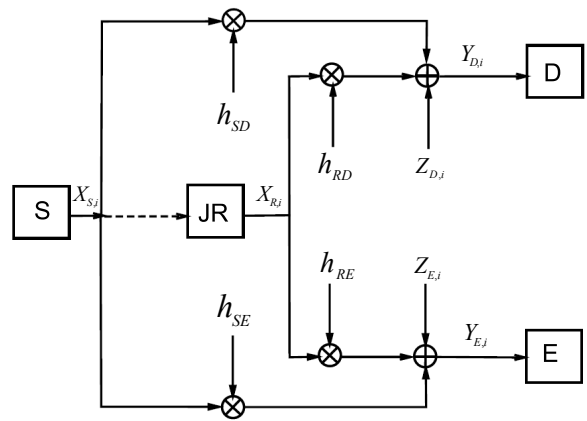}
\caption{The system model shows the source (S), the destination (D),
the eavesdropper (E) and the jammer relay (JR). The jammer relay aims to assist the
eavesdropper. The S-JR link is shown dashed, as the jammer relay does not know the source signal in the first model, whereas it is given the source signal causally in the second.} \label{fig:systemmodel}
\end{center}
\end{figure}

\begin{figure}[t]
\centering
\includegraphics[width=8cm]{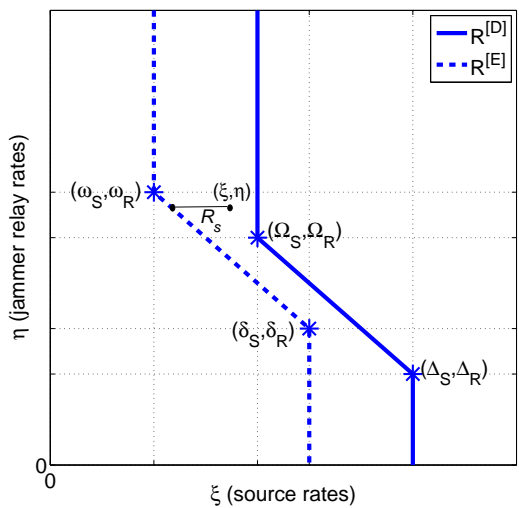}
\caption{Problem 1: Boundaries for $\mathcal{R}^{[D]}$ and $\mathcal{R}^{[E]}$, defined respectively in (\ref{eqn:R-D}) and (\ref{eqn:R-E}), under conditions (\ref{eqn:cond1})-(\ref{eqn:cond5}).}
\label{fig:DifferentCases}
\end{figure}

\begin{figure}[t]
\begin{center}
\includegraphics[width= 8cm]{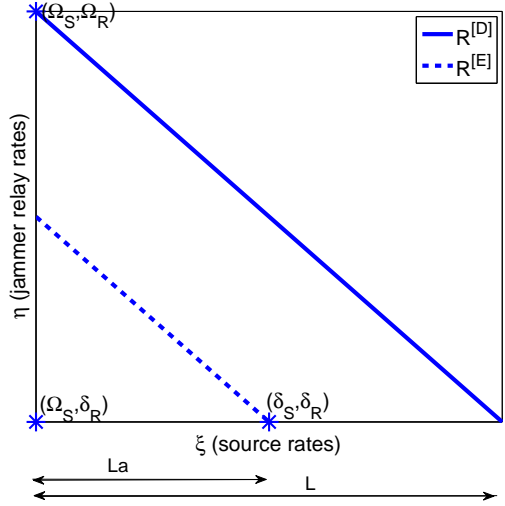}
\caption{The equivalent game for the game in Fig.~\ref{fig:DifferentCases}.} \label{fig:ReducedCaseA}
\end{center}
\end{figure}

\begin{figure}
\centering
\includegraphics[width = 8cm]{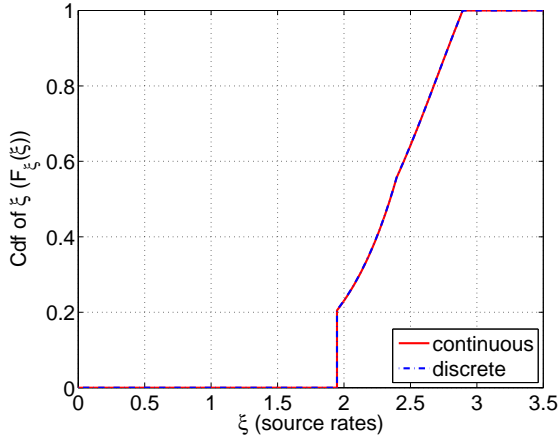}
\caption{Problem 1: Optimal cdfs for the source for continuous and discrete games when $a = 0.5255$, where $a$ is defined in Theorem \ref{thm:Solution}.}
\label{fig:CDFcaseAa1}
\end{figure}

\begin{figure*}[h]
\centering
\subfigure[Case A]{
\includegraphics[scale=0.45]{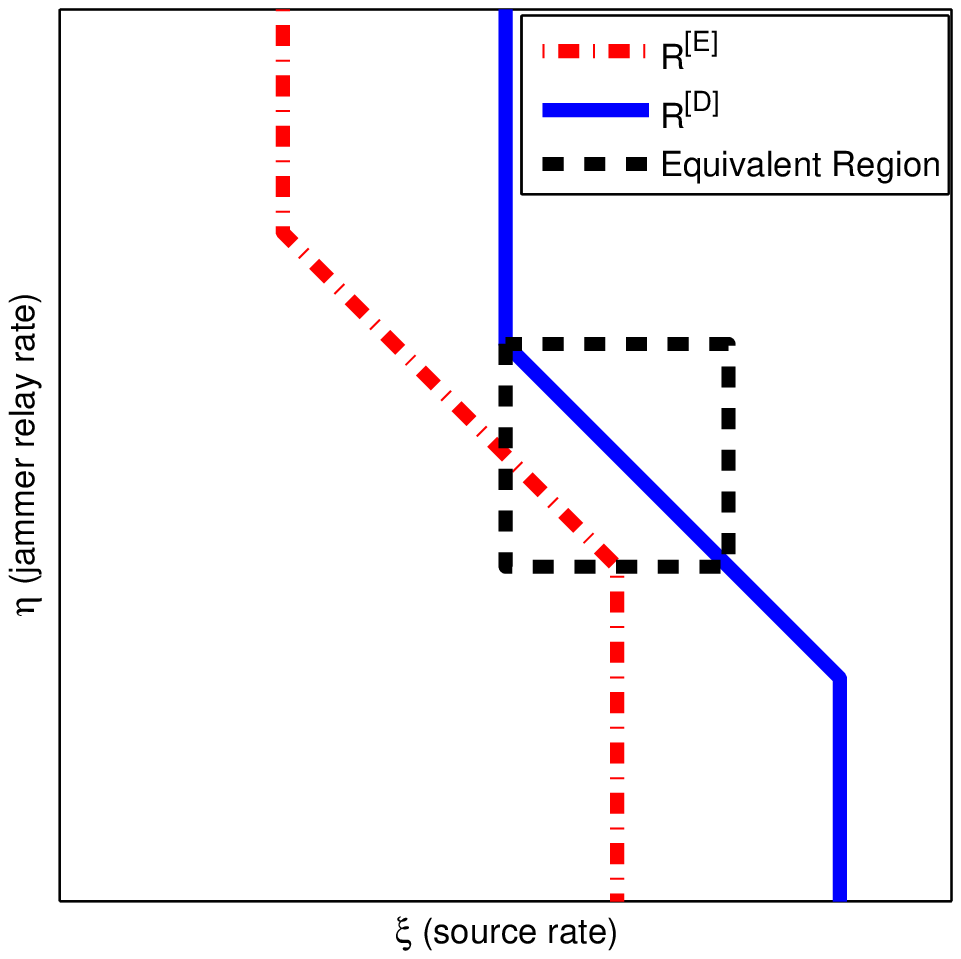}
\label{fig:subfig1}
}
\hspace{-1.6cm}
\subfigure[Case B]{
\includegraphics[scale=0.45]{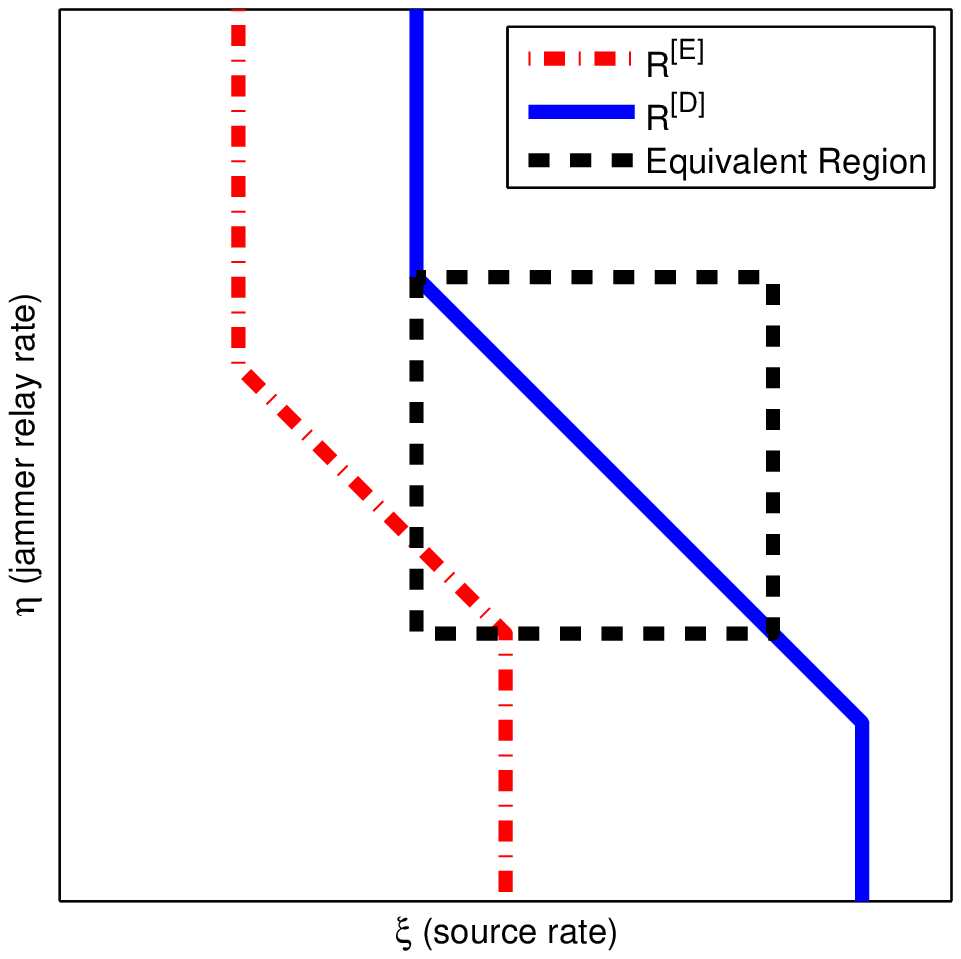}
\label{fig:subfig2}
}
\hspace{-1.6cm}
\subfigure[Case C]{
\includegraphics[scale=0.45]{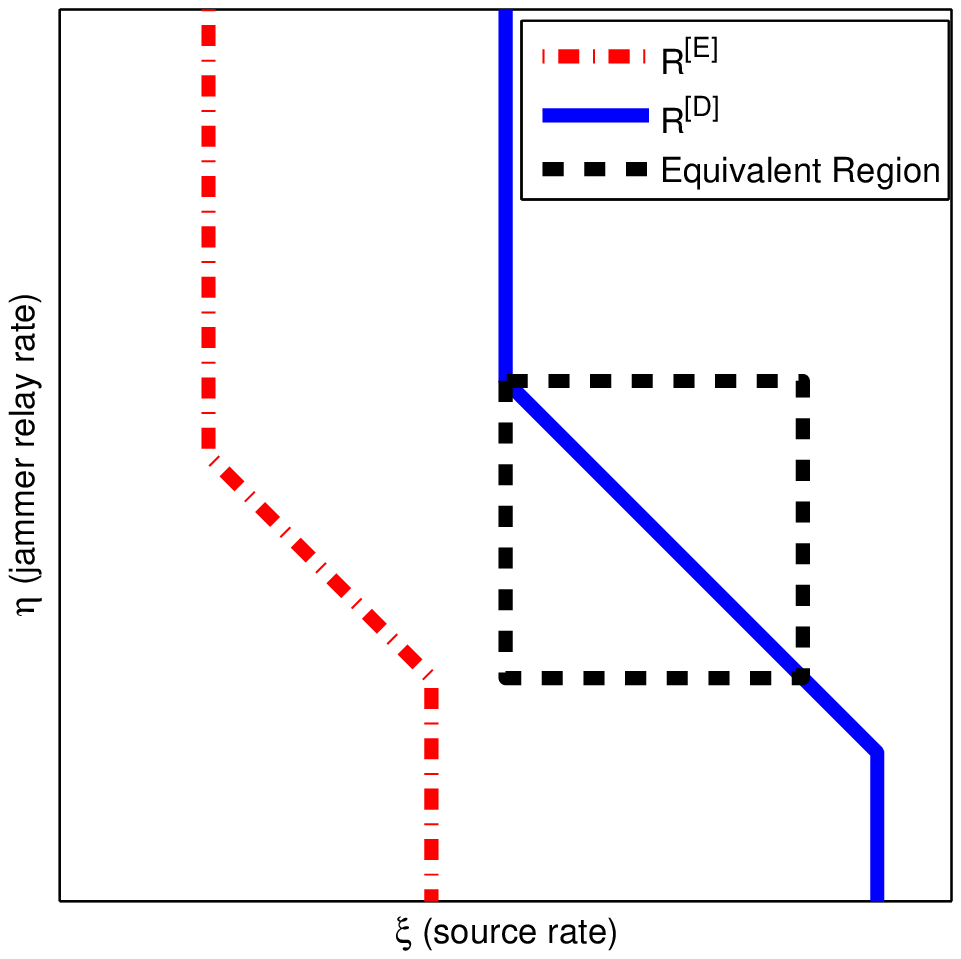}
\label{fig:subfig3}
}
\hspace{-1.6cm}
\subfigure[Case D]{
\includegraphics[scale=0.45]{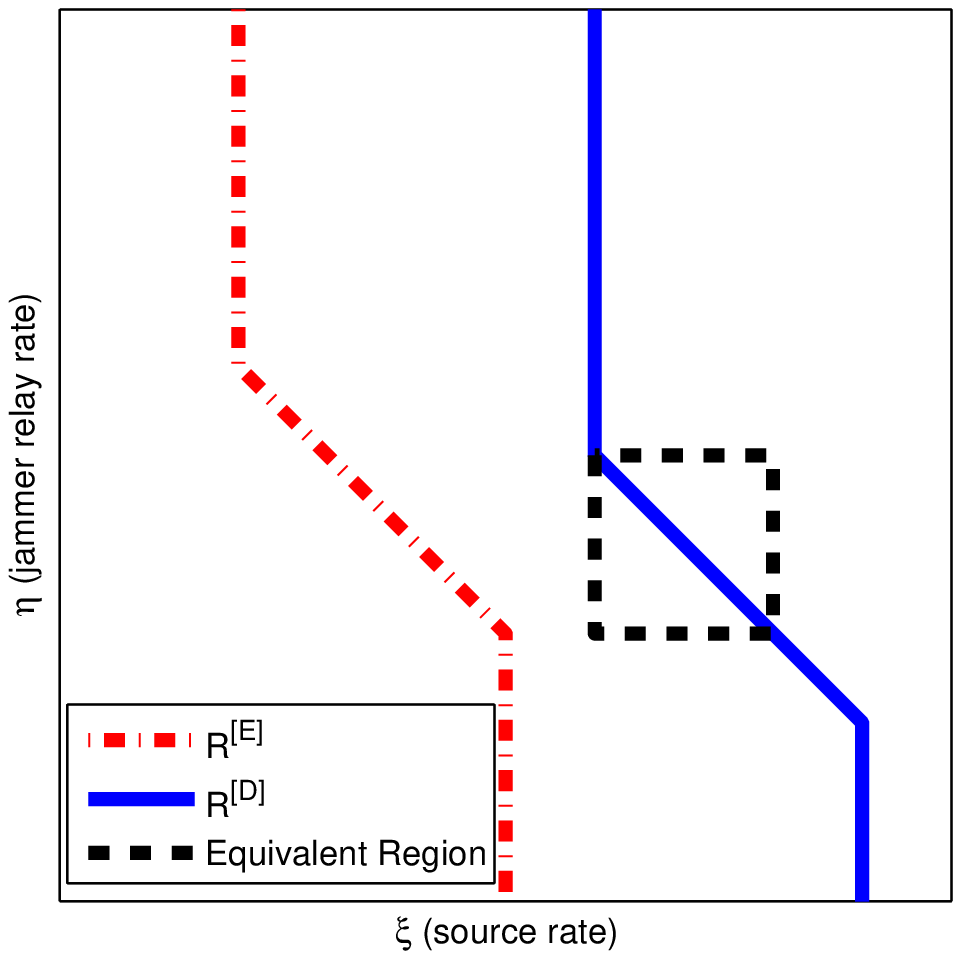}
\label{fig:subfig4}
}
\hspace{-1.6cm}
\subfigure[Case E]{
\includegraphics[scale=0.45]{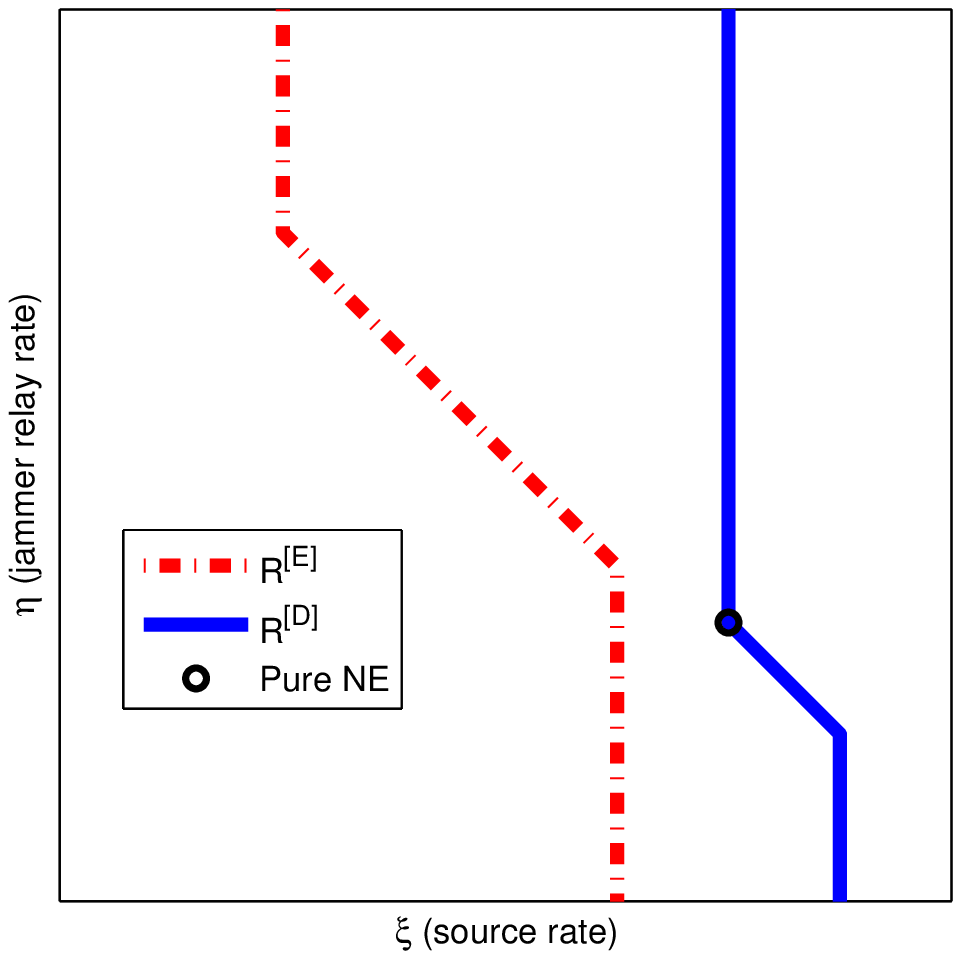}
\label{fig:subfig5}
}
\hspace{-1.6cm}
\subfigure[Case F]{
\includegraphics[scale=0.45]{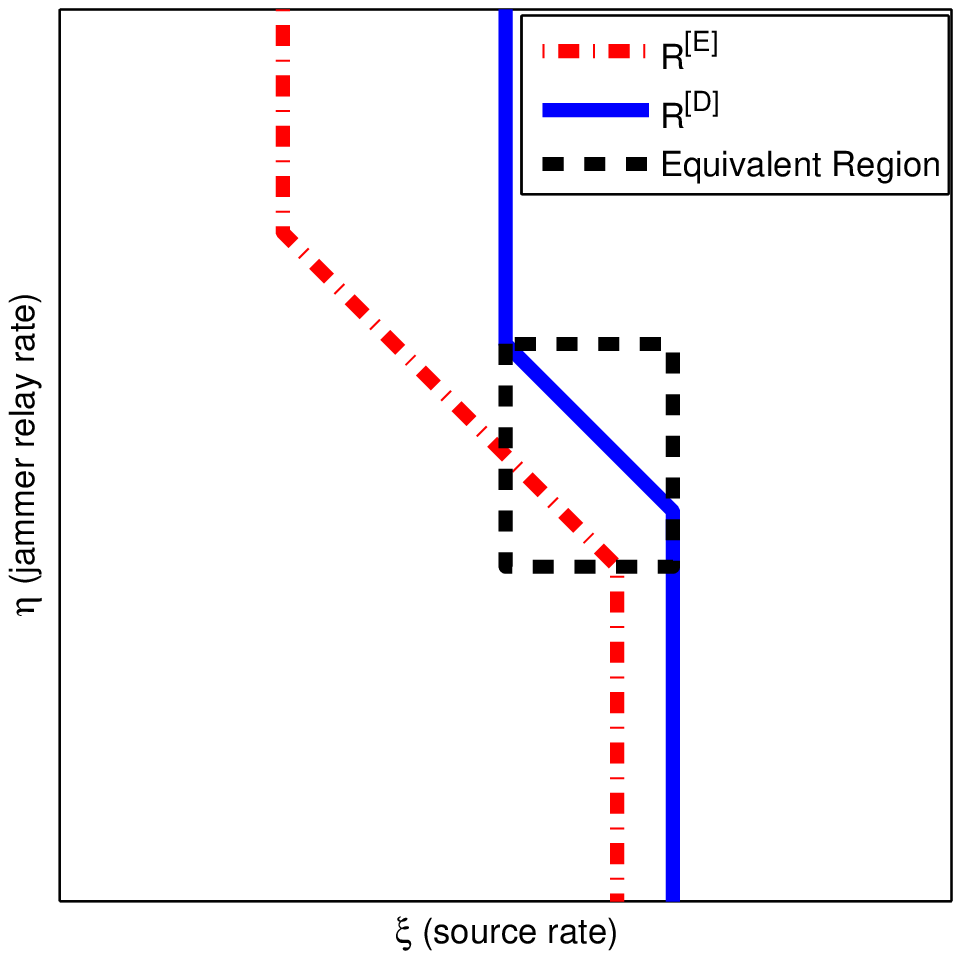}
\label{fig:subfig6}
}
\hspace{-1.6cm}
\subfigure[Case G]{
\includegraphics[scale=0.45]{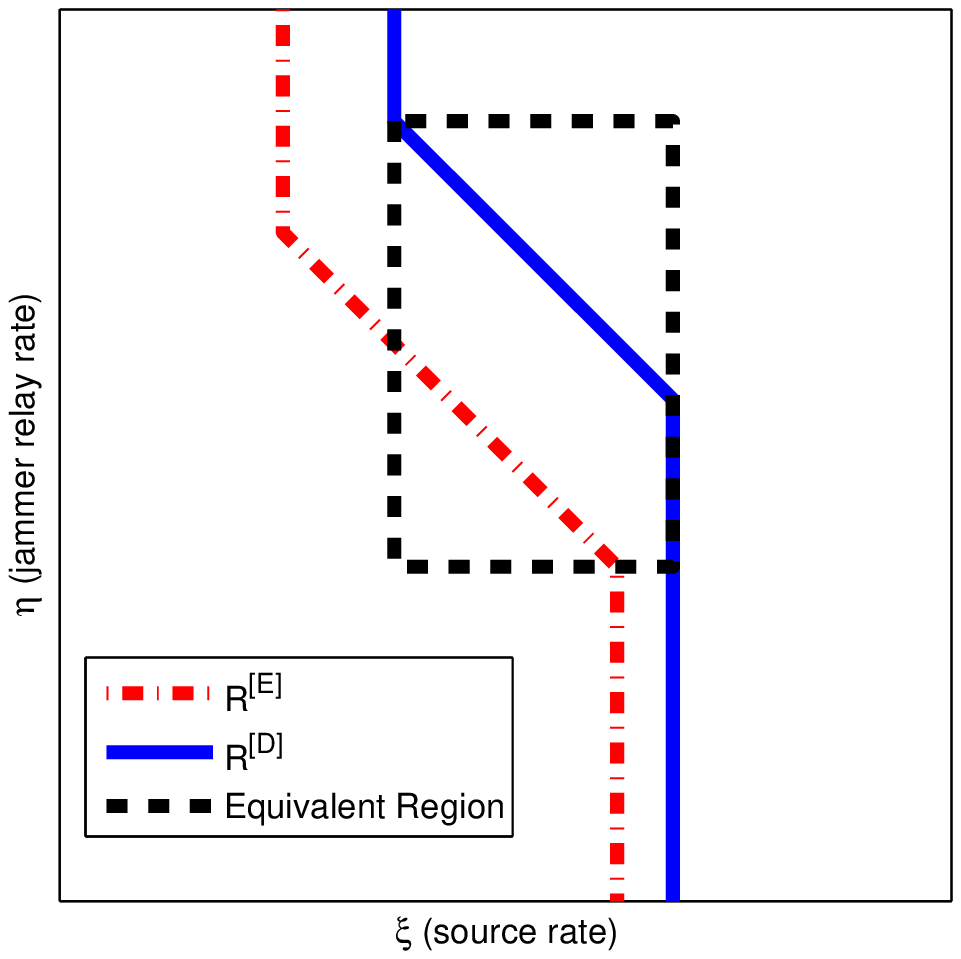}
\label{fig:subfig7}
}
\hspace{-1.6cm}
\subfigure[Case H]{
\includegraphics[scale=0.45]{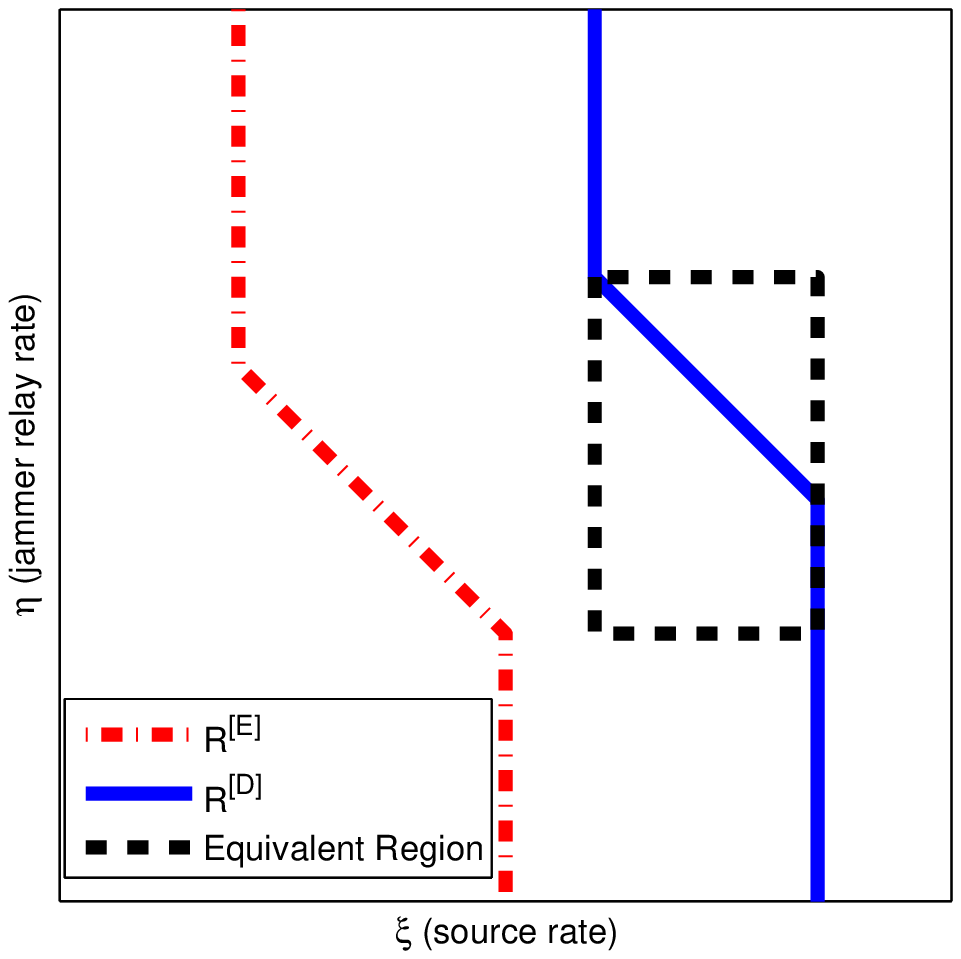}
\label{fig:subfig8}
}
\hspace{-1.6cm}
\subfigure[Case K]{
\includegraphics[scale=0.45]{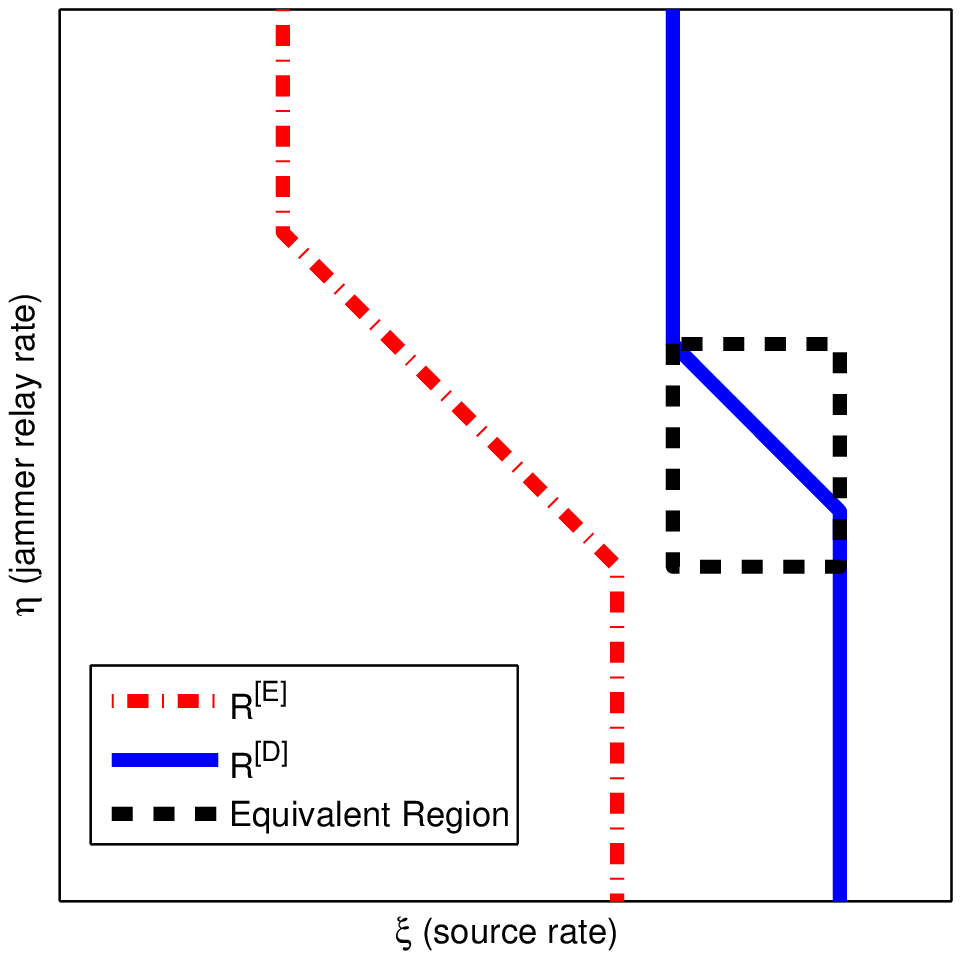}
\label{fig:subfig9}
}
\hspace{-1.6cm}
\subfigure[Case L]{
\includegraphics[scale=0.45]{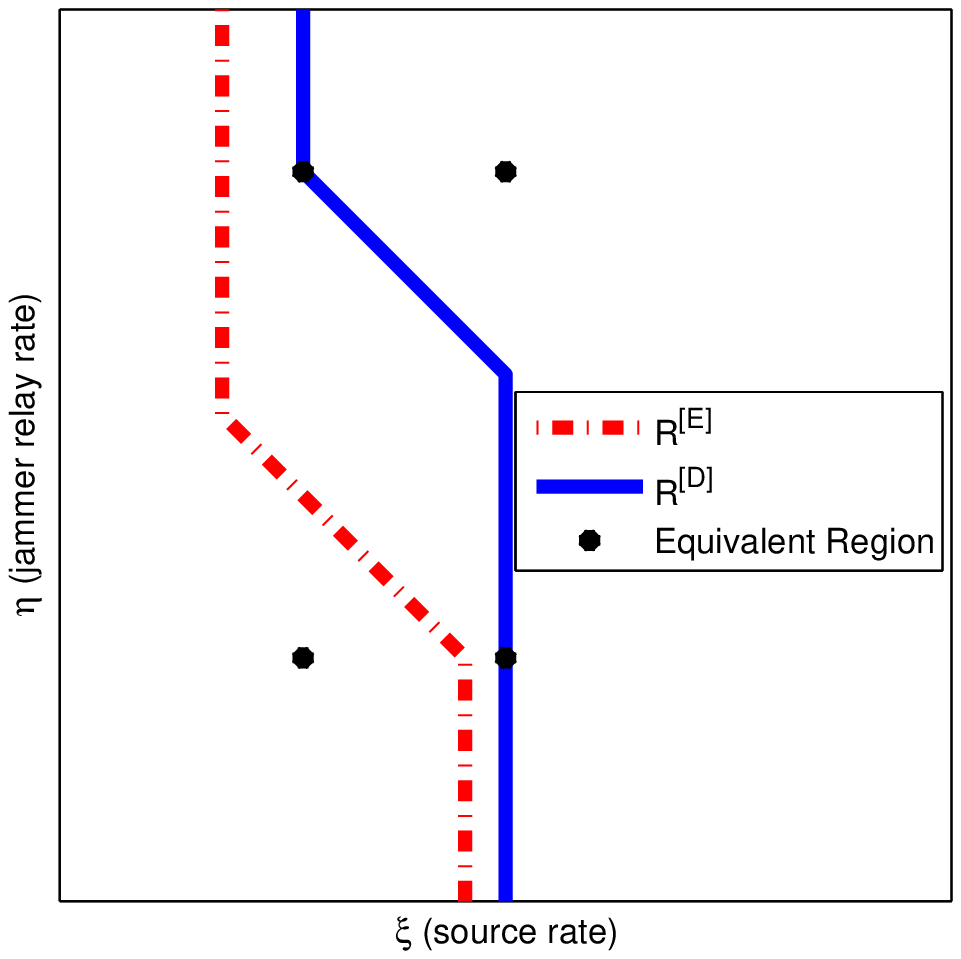}
\label{fig:subfig10}
}
\hspace{-1.6cm}
\subfigure[Case M]{
\includegraphics[scale=0.45]{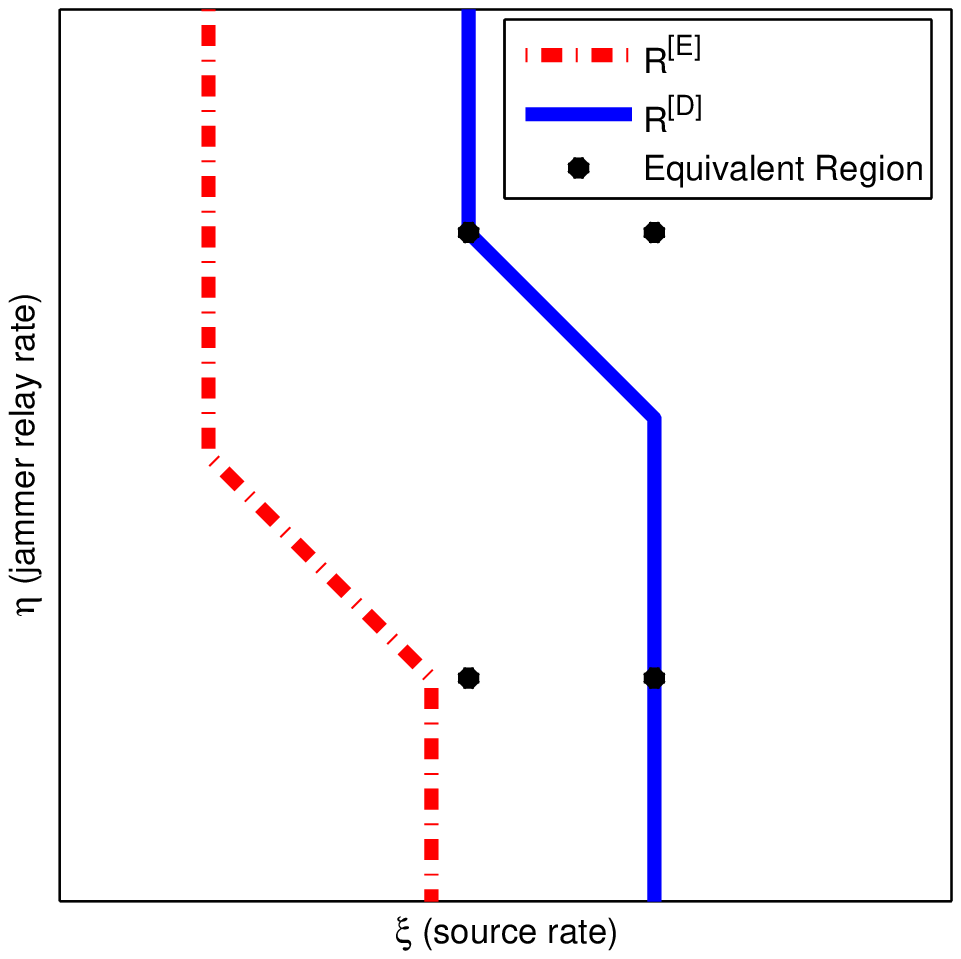}
\label{fig:subfig11}
}
\caption{Problem 1: Cases A-M. Case N, which is not shown in the figure, is the case when $\mathcal{R}^{[D]}$ and $\mathcal{R}^{[E]}$, respectively defined in (\ref{eqn:R-D}) and (\ref{eqn:R-E}), intersect or $\mathcal{R}^{[D]}$ is contained in $\mathcal{R}^{[E]}$.}\label{fig:casea_m}
\end{figure*}

\begin{figure}[t]
\centering
\includegraphics[width = 8cm]{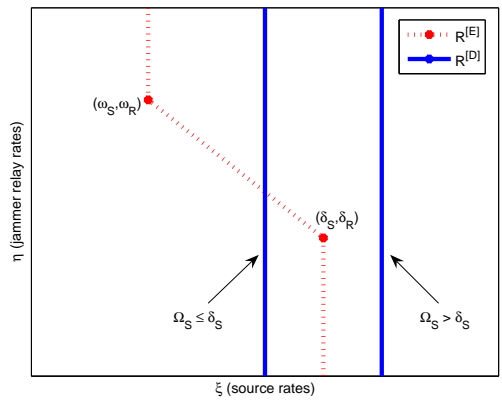}
\caption{Problem 1: boundary regions when the destination does not know the jammer relay codebook.}\label{fig:region_combined}
\end{figure}

\begin{figure}[t]
\centering
\includegraphics[width = 8cm]{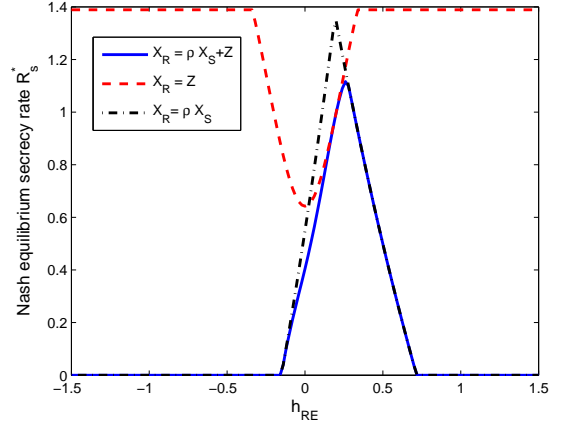}
\caption{Problem 2: Nash equilibrium secrecy rate as a function of $h_{RE}$ when $h_{SD} = 1$, $h_{SE} = 0.4+0.4j$ and $h_{RD} = 0.2-0.2j$. The signal $Z_1^n$, defined in (\ref{eqn:JRstrategy}), is Gaussian noise.}
\label{fig:Fig2}
\end{figure}

\begin{figure}
\centering
\includegraphics[width=8cm]{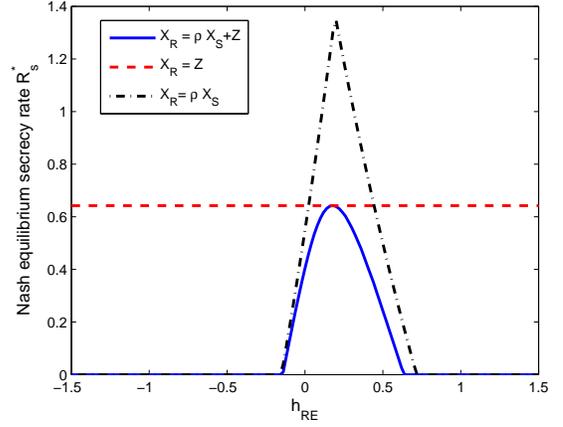}
\caption{Problem 2: Nash equilibrium secrecy rate as a function of $h_{RE}$ when $h_{SD} = 1$, $h_{SE} = 0.4+0.4j$ and $h_{RD} = 0.2-0.2j$. The signal $Z_1^n$, defined in (\ref{eqn:JRstrategy}), is a structured codeword.}
\label{fig:Fig6}
\end{figure}

\begin{figure}
\centering
\includegraphics[width= 8cm]{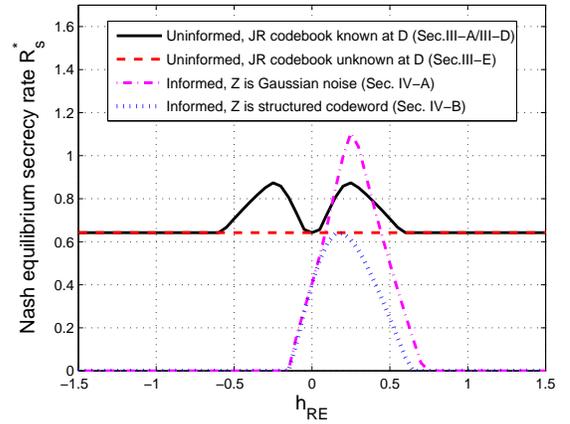}
\caption{Comparison of Nash equilibrium secrecy rates under different scenarios as a function of real $h_{RE}$ when $h_{SD} = 1$, $h_{SE} = 0.4+0.4j$ and $h_{RD} = 0.2-0.2j$.}
\label{fig:Fig7}
\end{figure}

\begin{figure}
\centering
\includegraphics[width= 8cm]{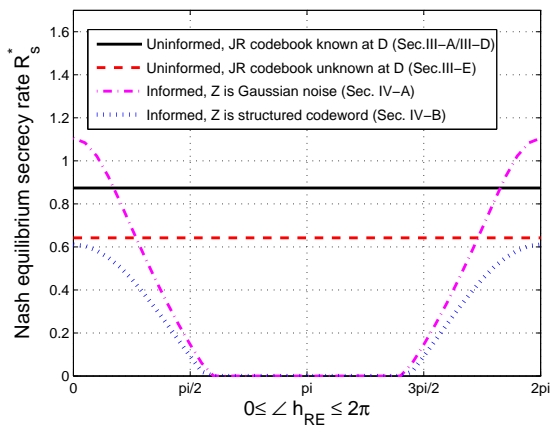}
\caption{Comparison of Nash equilibrium secrecy rates under different scenarios as a function of $\angle h_{RE}$ when $|h_{RE}| = 0.25$, $h_{SD} = 1$, $h_{SE} = 0.4+0.4j$ and $h_{RD} = 0.2-0.2j$.}
\label{fig:Fig8}
\end{figure}

\begin{figure}[t]
\begin{center}
\includegraphics[width= 8cm]{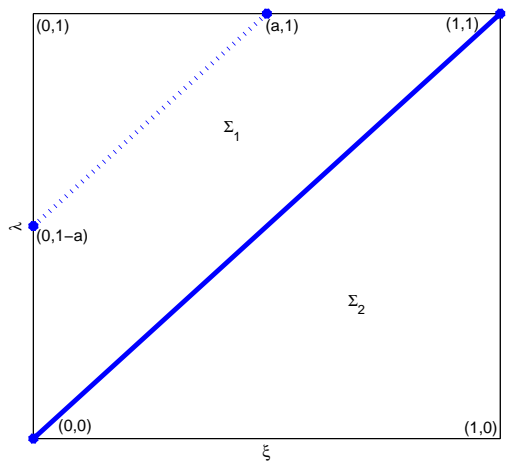}
\caption{Problem 1: the simplified game over the unit square} \label{fig:unitsquare}
\end{center}
\end{figure}

\end{document}